%% file: TrainMPC.tex
\documentclass[lettersize,journal]{IEEEtran}
\usepackage{amsmath,amsfonts}
\usepackage{algorithmic}
\usepackage{algorithm}
\usepackage{array}
\usepackage[caption=false,font=normalsize,labelfont=sf,textfont=sf]{subfig}
\usepackage{textcomp}
\usepackage{stfloats}
\usepackage{url}
\usepackage{verbatim}
\usepackage{graphicx}
\usepackage{cite}
\usepackage{color}
\usepackage{gensymb}
\usepackage{amsthm}
\usepackage{hyperref}

\usepackage{tikz}
\usetikzlibrary{tikzmark}
\usetikzlibrary{patterns}
\usetikzlibrary{shapes,arrows}
\usetikzlibrary{chains,matrix,positioning,scopes,calc}
\usetikzlibrary{shapes.geometric,intersections,backgrounds,fit}
\usetikzlibrary{decorations.text,decorations.pathreplacing}
\usepackage{circuitikz}
\usepackage{xcolor}
\tikzset{
	block/.style = {draw, rectangle,
		minimum height=1cm,
		minimum width=2cm},
	input/.style = {coordinate,node distance=1cm},
	output/.style = {coordinate,node distance=4cm},
	arrow/.style={draw, -latex,node distance=2cm},
	pinstyle/.style = {pin edge={latex-, black,node distance=2cm}},
	sum/.style = {draw, circle, node distance=1cm},
}

\newcommand{\INPUT}{\item[\textbf{Input:}]}
\newcommand{\OUTPUT}{\item[\textbf{Output:}]}

\newtheorem{theorem}{Theorem}

\begin{document}

\title{A Data-driven Predictive Control Architecture for Train Thermal Energy Management}

\author{
Ahmed Aboudonia, Johannes Estermann, Keith Moffat, Manfred Morari, John Lygeros  
\thanks{*This work was supported by the Swiss National Science Foundation under NCCR Automation (Corresponding Author: Ahmed Aboudonia)}
\thanks{Ahmed Aboudonia, Keith Moffat, Manfred Morari and John Lygeros are with the Automatic Control Laboratory, Department of Information Technology and Electrical Engineering, ETH Zurich, 8092 Zurich, Switzerland (emails: {\tt\small \{ahmedab,kmoffat,morari,lygeros\} @ethz.ch}). Johannes Estermann is with the Swiss Federal Railways (email: {\tt\small \{johannes.estermann\} @sbb.ch})}%
}

\markboth{Journal of \LaTeX\ Class Files,~Vol.~14, No.~8, August~2021}%
{Shell \MakeLowercase{\textit{et al.}}: A Sample Article Using IEEEtran.cls for IEEE Journals}


\maketitle

\begin{abstract}
We aim to improve the energy efficiency of train climate control architectures, with a focus on a specific class of regional trains operating throughout Switzerland, especially in Zurich and Geneva. Heating, Ventilation, and Air Conditioning (HVAC) systems represent the second largest energy consumer in these trains after traction. The current architecture comprises a high-level rule-based controller and a low-level tracking controller. To improve train energy efficiency, we propose adding a middle data-driven predictive control layer aimed at minimizing HVAC energy consumption while maintaining passenger comfort. The scheme incorporates a multistep prediction model developed using real-world data collected from a limited number of train coaches. To validate the effectiveness of the proposed architecture, we conduct multiple experiments on a separate set of train coaches; our results suggest energy savings between $10\%$ and $35\%$ with respect to the current architecture.
\end{abstract}


\section{Introduction}
\label{sec:I}

\IEEEPARstart{C}{limate} considerations have motivated an extensive research effort on reducing energy consumption and improving energy efficiency of various systems such as building automation \cite{afram2014theory}, industrial processes \cite{abdelaziz2011review} and the railway  sector \cite{wang2011survey}.
In the railway sector, the majority of research has mainly focused on traction energy management. As pointed out in \cite{novak2018hierarchical}, different approaches have been developed to reduce the energy consumption of individual trains, ranging from speed trajectory optimization techniques \cite{lu2013single} and fuzzy predictive control \cite{bai2014energy} to pseudospectral methods and mixed integer linear programming \cite{wang2016optimal}. Additionally, the development of cooperative scheduling approaches \cite{yang2012cooperative,yang2015energy} and the installation of energy storage systems \cite{zhu2017hierarchical,clerici2018multiport} have contributed to improving the overall energy efficiency by making use of the regenerative braking energy.

Although traction is the main energy consumer in trains, it is not the only significant one. For example, the heating, ventilation, and air conditioning (HVAC) systems are the second-largest energy consumers in the trains of the Swiss Federal Railways (SBB), accounting for approximately $20\%$ to $40\%$ of the total energy consumption \cite{vetterli2015energy}.
Therefore, improving the energy efficiency of HVAC systems might lead to significant monetary and environmental benefits. 
Optimal control presents a promising direction for reducing train HVAC energy consumption.

Here we consider the Regio-Dosto train fleet operated by the Swiss Federal Railways Company (SBB) for short distance journeys especially in the Zurich and Geneva areas. A two-level control architecture is currently used to regulate the coach temperature of Regio-Dosto trains, with a high-level rule-based controller selecting an appropriate temperature setpoint and a low level controller tracking it. The driver of the train can intervene between the two layers and modify the desired setpoint by up to $\pm 2 \degree C$ to accommodate requests from the passengers. To improve the energy efficiency of the existing architecture while keeping the train temperature within prespecified acceptable bounds, we propose to add a mid-level predictive control layer that operates at the same level as the driver.

Model predictive control (MPC) is an optimal control scheme that computes control actions by repeatedly solving an optimization problem \cite{rawlings2017model,kouvaritakis2016model,lee2011model}. MPC has been used to control a wide range of complex systems thanks to its ability to systematically consider optimization requirements (in our case, minimizing energy consumption), incorporate forecast information (e.g. weather forecasts) and integrate state and input constraints (e.g. temperature constraints).

The primary difficulty in designing an effective MPC scheme lies in the necessity of having a model which might be difficult to obtain or complex to use. 
Therefore, significant efforts have focused on developing Data-driven Predictive Control (DDPC) schemes where either grey-box or black-box models are obtained using historical measurement data. 
In the energy management domain, robust MPC \cite{maasoumy2014handling,maasoumy2012optimal,gao2023energy}, stochastic MPC \cite{long2014scenario,zhang2014sample,ma2014stochastic,oldewurtel2013stochastic}, adaptive MPC \cite{herzog2013self,schmelas2017savings,tanaskovicrobust}, and distributed MPC \cite{lauro2014adaptive}, have been extensively explored for building climate control applications and have experimentally demonstrated success in improving performance. DDPC schemes have also been proposed in this sector, based on various machine learning algorithms (e.g. Gaussian processes \cite{maddalena2022experimental}, neural networks \cite{bunning2021input}, random forests \cite{bunning2020experimental}) as well as data-driven tools (e.g. behavioral system theory \cite{di2022lessons}, physics-informed autoregressive moving average models \cite{bunning2022physics}) and other statistical methods \cite{aswani2013provably}. These schemes have been employed in various settings and locations, including Hospitals \cite{maddalena2022experimental}, university facilities \cite{aswani2013provably,di2022lessons} and residential units \cite{bunning2020experimental,bunning2021input,bunning2022physics}.

MPC and DDPC approaches are less common for HVAC systems in the railway sector. Though superficially similar to building HVAC systems, train HVAC systems pose their own challenges due to rapidly changing disturbances, including occupancy, speed, direction, and the resulting variations in ambient temperature and solar radiation. \color{black} An MPC scheme was developed for tram climate control and tested in a climatic wind tunnel in \cite{hofstadter2018energy}. Other control techniques, such as sliding mode control \cite{shah2021advanced}, neural network-based control \cite{lepore2022neural}, passenger-centric control \cite{buonomano2024energy} and model-based optimal control \cite{kumar2010design} have also been proposed though successful validation in experiments and uptake remains limited.

To improve the energy efficiency of the existing control architecture in the Regio-Dosto trains, we propose adding a DDPC layer that decides the optimal setpoint taking both energy efficiency and passenger comfort into account. In this layer, a data-driven multistep prediction model is used and derived based on the Transient Predictor \cite{moffat2024transient} that minimizes the 2-norm of the training data, single-step prediction error for each of the future timesteps in a computationally efficient manner using the LQ Decomposition. 
We show that the Transient Predictor-based DDPC layer can be designed for the whole fleet using a small amount of data collected from a few coaches in different parts of Zurich.

The proposed layer is not meant to replace but complement the existing control architecture. More specifically, this layer receives the original setpoint from the rule-based controller and sends a modified setpoint to the tracking controller.
The rule-based controller remains in place, because its original setpoint determines the temperature bounds used inside the proposed layer.
Likewise, the tracking controller embedded inside the HVAC system manufactured by a third party is still used to control all underlying processes and functions of the HVAC system. We demonstrate the impact of the proposed architecture on experiments on several Regio Dosto trains of the SBB fleet. Though the coaches are unoccupied for safety considerations, to the best of the author's knowledge, this is the first deployment of a DDPC architecture on real world trains.

The contributions of this work are summarized as follows:
\begin{enumerate}

    \item Learning the thermal dynamics of train coaches using a limited amount of data collected from a few trains in Zurich.

    \item Designing the DDPC layer to determine the optimal energy-efficient temperature setpoint online.

    \item Evaluating the closed-loop performance and energy efficiency of the developed control architecture by conducting multiple experiments on various trains in different regions in Zurich.

\end{enumerate}

The paper is organized as follows. 
In Section \ref{sec:II}, we introduce the SBB simulator used to simulate the closed-loop thermal dynamics of the Regio-Dosto trains under the existing control architecture.
In Section \ref{sec:III}, we present the developed DDPC architecture. 
We start by analyzing the real data collected by SBB to develop this architecture in Section \ref{sec:IIIA}. 
We then use the collected data to design a data-driven multistep prediction model for the train thermal dynamics based on the Transient Predictor in Section \ref{sec:IIIB}. 
Finally, we formulate the optimal control problem of the proposed DDPC layer in Section \ref{sec:IIIC}. 
In Section \ref{sec:IV}, we run multiple simulations to evaluate the data driven  multistep prediction model and the DDPC architecture. 
In Section \ref{sec:V}, we report on the experiments \color{black} to investigate the energy efficiency of the proposed DDPC architecture on real trains. 
Finally, we provide some concluding remarks in Section \ref{sec:VI}.

\section{Thermal Dynamics Simulator}
\label{sec:II}

A Regio-Dosto train comprises multiple coaches, each of which is divided into two halves: right and left. Each coach is also equiped with two HVAC systems, with one system serving each half. Figure \ref{fig:coach} provides a schematic representation of one half of a coach. The figure illustrates that each coach has three levels: a lower deck, a middle deck and an upper deck. While the entire coach receives the same temperature setpoint, each HVAC system is designed to control each deck individually.

\begin{figure}[!t]
    \centering
    \includegraphics[trim=2cm 4cm 2cm 4cm,clip=true,width=8cm]{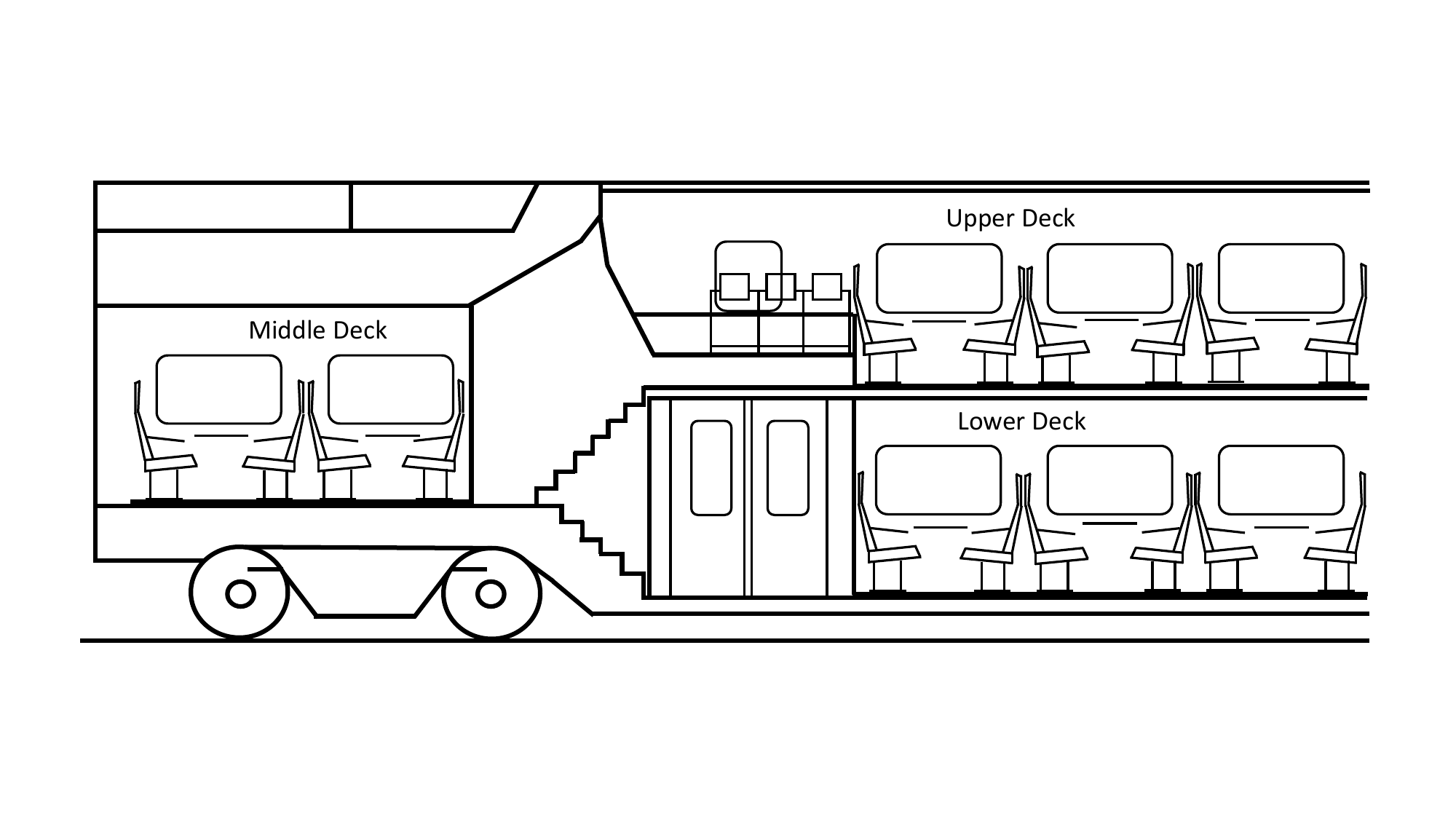}
    \caption{A schematic diagram of the considered coach.}
    \label{fig:coach}
\end{figure}

The SBB train thermal dynamics simulator, described in Figure \ref{fig:topology} models the closed-loop behavior of a single HVAC system along with its associated half coach under the existing control architecture. Though simplified, the SBB simulator still provides a good representation of the thermal dynamics and can be used both for testing the performance of the developed architecture in simulation before being implemented on a real train (Section VI) and for the implementation of the DDPC architecture (Section III.C).
An early description of the SBB simulator is provided in \cite{gasser2019Mach}.
As shown in the schematic diagram in Figure \ref{fig:topology}, the SBB simulator comprises three main blocks: the \textit{rule-based control block}, the \textit{HVAC system and control block}, and the \textit{train thermal dynamics block}. The forthcoming subsections describe each block.

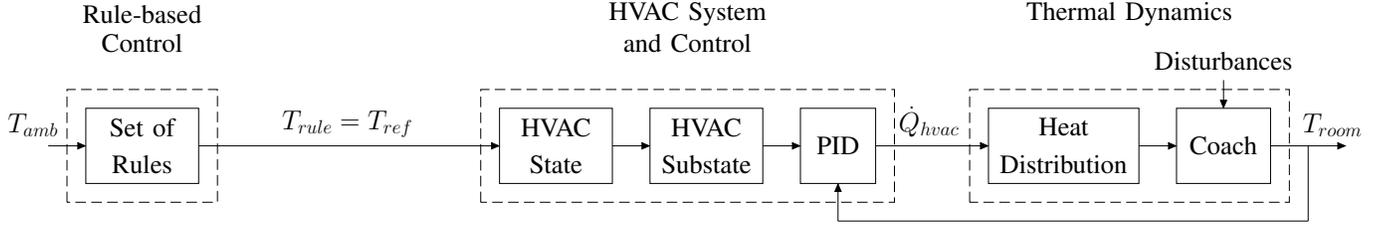
\begin{figure*}[!t]
	\centering
	\scalebox{0.5}{\input{Simulator.tex}}
	\caption{A block diagram of the SBB simulator.}
	\label{fig:topology}
\end{figure*}

\subsection{Rule-based Controller}

Based on a set of prespecified rules, the rule-based control block uses the ambient temperature $T_{amb}$ to determine the rule-based temperature $T_{rule}$. The reference temperature $T_{ref}$ is then set to the rule-based temperature $T_{rule}$. Although the HVAC system has a single setpoint $T_{ref}$, it can control each deck individually. For each HVAC system, there are three main modes: the \textit{regular mode}, the \textit{slumber mode} and the \textit{off mode}, each with its own prespecified rules. The set of prespecified rules leads to a non-decreasing piecewise linear relation between the reference temperature and the ambient temperature.

\subsection{HVAC System \& Control}

The HVAC system and control block contains the logic behind operating the HVAC system. This block takes the temperature setpoint $T_{ref}$ as an input and generates the required input heat flow rate $\dot{Q}_{hvac}$ so that the room temperature $T_{room}$ reaches the setpoint $T_{ref}$. Although all three decks share the same setpoint $T_{ref}$, each deck is controlled individually and hence, we define $T_{room} = [T_{room}^{up} \ T_{room}^{mid} \ T_{room}^{low}]^\top$ and $\dot{Q}_{hvac} = [\dot{Q}_{hvac}^{up} \ \ \dot{Q}_{hvac}^{mid} \ \ \dot{Q}_{hvac}^{low}]^\top$ where the superscripts $up$, $mid$ and $low$ refer, respecetively, to the upper, middle and lower decks shown in Figure \ref{fig:coach}.

The HVAC System and Control block mainly comprises three main sub-blocks: the HVAC state block, the HVAC substate block and the PID control block. The HVAC state block determines the HVAC state based on the HVAC mode as well as the room, reference and ambient temperatures.
The regular mode contains multiple states including preheating, precooling, heating, cooling and a mixed state where both heating and cooling could be activated. The slumber mode also includes multiple states including active heating, active ventilation, active cooling and an off state where the HVAC system is switched off.
The off mode only includes a single off state.
The HVAC substate block then decides the HVAC substate based on the HVAC state, the room temperature, the ambient temperature and the occupancy. 
The HVAC substate mainly determines whether circulated air or outside air or a mixture of both should be used by the HVAC system. 

Once the state and substate of each HVAC system are determined, a PID control law given by,
\begin{align*}
    \dot{Q}_{hvac}^{up}(t) &= k_p^{u} \delta T_{room}^{up}(t) + k_i^{u} \hspace{-0.15cm} \int \hspace{-0.15cm} \delta T_{room}^{up}(t) dt + k_d^{u} \delta \dot{T}_{room}^{up}(t) \\
    \dot{Q}_{hvac}^{mid}(t) &= k_p^{m} \delta T_{room}^{mid}(t) + k_i^{m} \hspace{-0.15cm} \int \hspace{-0.15cm} \delta T_{room}^{mid}(t) dt + k_d^{m} \delta \dot{T}_{room}^{mid}(t) \\
    \dot{Q}_{hvac}^{low}(t) &= k_p^{l} \delta T_{room}^{low}(t) + k_i^{l} \hspace{-0.15cm} \int \hspace{-0.15cm} \delta T_{room}^{low}(t) dt + k_d^{l} \delta \dot{T}_{room}^{low}(t)
\end{align*}
is implemented by the PID sub-block to drive the room temperature $T_{room}$ to the desired setpoint $T_{ref}$. Here $k_p^u$, $k_i^u$, $k_d^u$, $k_p^m$, $k_i^m$, $k_d^m$, $k_p^l$, $k_i^l$ and $k_d^l$ are the control gains and $\delta T_{room}^i(t) = T_{ref}^i(t) - T_{room}^i(t)$ denotes the error between the reference and room temperatures at time $t$ for the corresponding deck $i \in \{up,\ mid,\ low\}$. The HVAC state and substate determines the parameters of the PID controllers including the limits on the PID output and its rate of change, the PID initial conditions, the PID reset signal and the anti-windup feature.

Besides the main HVAC system, Regio Dosto trains are also equipped with a floor and wall heating system that can control each deck individually. Floor and wall heating is only activated in the regular mode during the preheating and heating states. The input heat flow rate of this system is denoted by $\dot{Q}_{fw} = [\dot{Q}_{fw}^{up} \ \dot{Q}_{fw}^{mid} \ \dot{Q}_{fw}^{low}]^\top$. While $\dot{Q}_{fw}$ is set to its maximum value in the preheating state, it is computed as a function of the ambient temperature using a piecewise linear relation in the heating state. 

\subsection{Thermal Dynamics}

This block includes a thermal model that describes the thermal dynamics of one half of a coach in the Regio-Dosto trains. It takes the heat flow rate of the HVAC system of the three decks $\dot{Q}_{hvac} = \left[ \dot{Q}_{hvac}^{up} \ \dot{Q}_{hvac}^{mid} \ \dot{Q}_{hvac}^{low} \right]^\top$ and that of the floor and wall heating system of the three decks $\dot{Q}_{fw}  = \left[ \dot{Q}_{fw}^{up} \ \dot{Q}_{fw}^{mid} \ \dot{Q}_{fw}^{low} \right]^\top$ as inputs and models the behavior of the room temperature $T_{room} = \left[ T_{room}^{up} \ T_{room}^{mid} \ T_{room}^{low} \right]^\top$. Overall, this model has a six dimensional input vector $\dot{Q}_{in} = \left[ \dot{Q}_{hvac} \ \dot{Q}_{fw} \right]^\top$ and a three dimensional output vector $T_{room} = \left[ T_{room}^{up} \ T_{room}^{mid} \ T_{room}^{low} \right]^\top$. 
The model comprises nine states which include the room temperature of the three decks $T_{room} = \left[ T_{room}^{up} \ T_{room}^{mid} \ T_{room}^{low} \right]^\top$, the inventory (seats, upholstery, etc) temperature of the three decks $T_{inv} = \left[ T_{inv}^{up} \ T_{inv}^{mid} \ T_{inv}^{low} \right]^\top$, and the chassis temperature $T_{chassis} = \left[ T_{chassis}^{up} \ T_{chassis}^{mid} \ T_{chassis}^{low} \right]^\top$ of the ceiling of the coach $\left( T_{chassis}^{up} \right)$, the floor of the coach $\left( T_{chassis}^{low} \right)$ and the floor between the upper and lower decks $\left( T_{chassis}^{mid} \right)$. 

Figure \ref{fig:ThermalDynamics} shows a schematic diagram of the heat transfer through air movement and the heat flow through conduction, convection and radiation among the sections of the train as well as between the train and the environment. 
The environment here refers to various factors such as the ambient temperature, the solar and ground radiations, the occupancy and the door effect. 

\begin{figure*}[ht]
	\centering
    \scalebox{0.5}{\input{ThermalDynamics.tex}}
	\caption{Thermal Dynamics \cite{gasser2019Mach}.}
	\label{fig:ThermalDynamics}
\end{figure*}
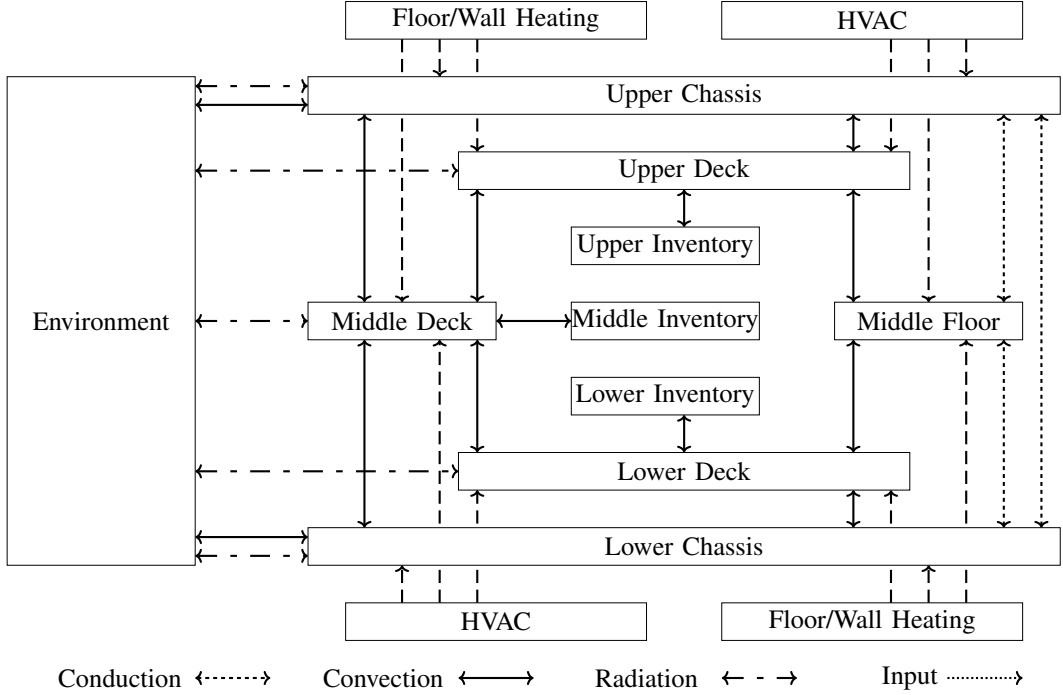

As shown in Figure \ref{fig:topology}, the thermal dynamics block is divided into two sub-blocks: the heat distribution block and the coach block.
The control inputs from each HVAC system and floor/wall heating system are not fully transferred to the corresponding deck. 
Indeed, a considerable part of the heat input is lost in the chassis of the train on its way to the deck. 
A small part of the heat input to each deck also affects the other decks. 
For this reason, the input vector $\dot{Q}_{in}$ is not the actual input used to regulate the room temperature $T_{room}$.
We denote the actual six dimensional input vector by $\dot{Q}_{act} = \left[ \dot{Q}_{room}^{up} \ \dot{Q}_{room}^{mid} \ \dot{Q}_{room}^{low} \ \dot{Q}_{chassis}^{up} \ \dot{Q}_{chassis}^{mid} \ \dot{Q}_{chassis}^{low} \right]^\top$ with $\dot{Q}_{room}^{up}$,  $\dot{Q}_{room}^{mid}$ and $\dot{Q}_{room}^{low}$ referring to the heat transferred to the upper, middle and lower decks, $\dot{Q}_{chassis}^{up}$ and $\dot{Q}_{chassis}^{low}$ denoting the heat transferred to the upper and lower parts of the chassis and $\dot{Q}_{chassis}^{mid}$ being the heat transferred to the middle floor between the upper and lower decks. 
The heat distribution block describes the relation between $\dot{Q}_{in}$ and $\dot{Q}_{act}$ that is given by,
\begin{equation*}
    \dot{Q}_{act} = \Lambda \dot{Q}_{in},
\end{equation*}
where $\Lambda$ is a matrix with each column summing up to 1 and whose $i$-th row and $j$-th column entry describes the fraction of the contribution of the $j$-th element in $\dot{Q}_{in}$ to the $i$-th element in $\dot{Q}_{act}$.

The coach sub-block models the thermal behavior of each part in the coach illustrated in Figure \ref{fig:ThermalDynamics}. The thermal behavior of each part is described using the generic relation,
\begin{equation}
    \label{eq:Qdot}
    \dot{T}_{p}^{d} = \frac{1}{m_{p}^{d} c_{p}^{d}} \sum_{i=1}^{N_p^d} \dot{Q}_{p,i}^{d}
\end{equation}
where $p \in \{room, \ chassis, \ inv\}$ refers to the part of the coach, $d \in \{up, \ mid, \ low\}$ refers to the deck, $m_p^d$ is the mass of part $p$ in deck $d$, $c_p^d$ is the specific heat capacity of part $p$ in deck $d$, $Q_{p,i}^{d}$ is the $i$-th heat flow rate affecting part $p$ in deck $d$ and $N_{p}^{d}$ is the number of heat flow rates affecting part $p$ in deck $d$. The heat flow rates $\dot{Q}_{p,i}^{d}$ result from various heat transfer mechanisms, such as convection, conduction, and radiation. In the sequel, we discuss the different heat flow rates in \eqref{eq:Qdot}.

\subsubsection{Input} Recall the input heat flow rate $\dot{Q}_{act} = \left[ \dot{Q}_{room}^{up} \ \dot{Q}_{room}^{mid} \ \dot{Q}_{room}^{low} \ \dot{Q}_{chassis}^{up} \ \dot{Q}_{chassis}^{mid} \ \dot{Q}_{chassis}^{low} \right]^\top$. For each $p \in \{room, \ chassis\}$ and $d \in \{ up, \ mid, \ low \}$, the input heat flow rate $\dot{Q}_{p}^{d}$ directly affects the rate of change of the corresponding temperature $T_{p}^{d}$.

\subsubsection{Convection} As shown in Figure \ref{fig:ThermalDynamics}, convection occurs among the decks, between the decks and the inventories, and between the decks and the chassis.
Convection yields a heat flow rate that depends on the temperature difference of the two parts involved and affects the rate of change of the temperature of both parts. 
The lower and upper chassis also exchange heat with the external environment due to convection.
This heat transfer leads to heat flow rates that are proportional to the difference between the ambient temperature on one side and the temperatures of the upper and lower chassis on the other side. 
Each of the resulting heat flow rates affects the rate of change of the temperature of the corresponding chassis. 
An additional complication is that the proportionality constant of the heat exchange with the environment depends on the train speed.

\subsubsection{Conduction} Conduction takes place between the upper chassis, the lower chassis and the middle floor as shown in Figure \ref{fig:ThermalDynamics}. Similar to convection, conduction results in a heat flow rate that depends on the temperature difference between the two parts involved and affects the rate of change of the temperature of both parts.

\subsubsection{Solar Radiation} Solar radiation plays a key role in the thermal dynamics of the train by affecting all decks through the windows. Additionally, it impacts both the upper and lower chassis. 
By defining the global irradiation $\dot{Q}_{G}$, the solar elevation angle $\alpha$, the solar azimuth angle $\beta$ and the train direction $\theta$ 
, the vertical component of the solar radiation affecting the upper chassis is proportional to $|\dot{Q}_G \cos(\alpha)|$.
Whereas the horizontal component of the solar radiation affecting the chassis and the decks is proportional to $|\dot{Q}_G \sin(\alpha) \sin(\beta-\theta)|$. Note that the simulator does not consider local effects, such as buildings, bridges and tunnels on solar radiation—an interesting topic for future investigation.

\subsubsection{Ground Radiation} This effect occurs due to the thermal radiation emitted by the ground as the track of the train absorbs heat from sunlight and radiates it. Denoting the track temperature by $T_G$, according to the model, the heat flow rate due to ground radiation occurs only when the ambient temperature $T_{amb}$ exceeds $20 \degree C$, depends on the difference $T_G^4 - T_{chassis}^{low^4}$ and affects the rate of change of the temperature $T_{chassis}^{low}$. The track temperature $T_G$ is computed as a function of the ambient temperature $T_{amb}$.

\subsubsection{Occupancy} Passengers could be a major heat source inside the train. To model their thermal effect, a heat flow rate of  100 W is added for each passenger inside the train. This heat source affects the rate of change of the temperature of the deck where the passenger is.

\subsubsection{Door Effect} As the train doors open when the train arrives in a station, heat is transferred between the train and the external environment. Since the doors are in the lower deck, the resulting heat flow rate affects the rate of change of the temperature $T_{room}^{low}$ and depends on the difference between the ambient temperature and the temperature of the lower deck.

\section{Methodology}
\label{sec:III}

In this section, we make use of data from real trains (Section \ref{sec:IIIA}) to present a data driven model (Section \ref{sec:IIIB}) as a basis for the DDPC architecture (Section \ref{sec:IIIC}). 

\subsection{Real Data Processing}
\label{sec:IIIA}

The SBB data for the Regio Dosto fleet used here comprises two datasets. The first is collected between February 22, 2022 and March 31, 2022 from eight coaches in eight different trains.
The second is collected between July 4, 2023 and July 23, 2023 from four coaches in four different trains.
In both datasets, the involved trains are operating in different routes under different weather conditions.
The trains used to collect the second dataset are a subset of the trains from which the first dataset is collected.

Each dataset contains raw data from three sources: HVAC data, weather measurements and train recordings. 
The HVAC data includes the HVAC mode, the reference temperature $T_{ref}$, the room temperature $T_{room}$ and the train speed $V$. The sampling time of the Winter/Spring dataset is 1 minute, whereas that of the Summer dataset is 2 seconds. 
The weather measurements include the train's latitude and longitude together with the corresponding ambient temperature $T_{amb}$, global irradiation $\dot{Q}_{G}$, solar elevation angle $\alpha$ and solar azimuth angle $\beta$. The sampling time of the Winter/Spring dataset is 10 minutes, whereas that of the Summer dataset is 1 minute.
We use the train's latitude and longitude to calculate the train's direction.
The train recordings include the schedule of each train as well as the occupancy data of each trip, represented as a percentage $P$ of the train's maximum capacity. Unlike the HVAC and weather data, the train recordings 
do not have a specific sampling time as it is recorded per trip.

Although the raw data is rich, it can not be used directly due to the missing data and the variability in the data formats and the sampling times.
For the data to be usable, we go through multiple data processing steps.
First, we fuse the data from the different sources to have a consistent and informative data trajectory for each train on each day with the same data format. 
We then aggregate the resulting trajectories to have a sampling time of $1$ minute for all trajectories.
Next, we segment and filter the resulting data to reach distinct trajectories preserving temporal continuity without any missing data.
The final step is a down-sampling step to generate filtered trajectories with sampling times of $5$ and $10$ minutes to allow DDPC at different granularities. Table \ref{tab:my_label} shows the total number of trajectories $\#T$ and the average number of data points $\#D$ in every trajectory for each combination of dataset and sampling time.

\begin{table}[]
    \centering
    \begin{tabular}{|c||cc|cc|cc|}
        \hline
        & \multicolumn{2}{c|}{1 minute} & \multicolumn{2}{c|}{5 minutes} & \multicolumn{2}{c|}{10 minutes} \\
        \cline{2-7}
        & $\#$T & $\#$D & $\#$T & $\#$D & $\#$T & $\#$D \\
        \hline
        Winter/Spring & 344 & 266 & 320 & 57 & 303 & 30 \\
        \hline
        Summer & 141 & 316 & 139 & 64 & 138 & 32 \\
        \hline
    \end{tabular}
    \caption{Processed data trajectories}
    \label{tab:my_label}
\end{table}

\subsection{Data-driven Prediction Model}
\label{sec:IIIB}

We denote the total number of trajectories resulting from the data processing steps in Section \ref{sec:IIIA} with a specific sampling time by $N$ and the length of the $i$-th trajectory by $n_i$ where $i \in \{1,\hdots,N\}$. The $i$-th trajectory comprises the input data points $u_i^l$, the output data points $y_i^l$ and the disturbance data points $d_i^l$ given by,
\begin{equation*}
    \begin{aligned}
        u_i^l &= T_{ref,i}^l \quad
        y_i^l = 
        \begin{bmatrix} 
            T_{room,i}^{up,l} &
            T_{room,i}^{mid,l} & 
            T_{room,i}^{low,l} 
        \end{bmatrix}^\top, \\
        d_i^l &= 
        \begin{bmatrix} 
            T_{amb,i}^l & 
            \dot{Q}_{Gt,i}^l & 
            \dot{Q}_{Gs,i}^l & 
            P_{i}^l & 
            V_{i}^l
        \end{bmatrix}^\top, 
    \end{aligned}
\end{equation*}
where $l \in \{1,\hdots,n_i\}$, $\dot{Q}_{Gt}^l = |\dot{Q}_G^l cos(\alpha^l)|$ and $\dot{Q}_{Gs}^l = |\dot{Q}_G^l sin(\alpha^l) sin(\beta^l - \theta^l)|$ are the projections of the solar radiation on the top and side of the train, respectively. Thus, for the $i$-th trajectory, we define the input $u_i$, output $y_i$ and disturbance $d_i$ as,
\begin{equation*}
    \begin{aligned}
        u_i = \left[u_i^{1^\top}, \hdots \right. & \left. u_i^{n_i^\top}\right]^\top, \quad
        y_i = \left[y_i^{1^\top}, \hdots y_i^{n_i^\top}\right]^\top, \\
        &d_i = \left[d_i^{1^\top}, \hdots d_i^{n_i^\top}\right]^\top,
    \end{aligned}
\end{equation*}
We use these trajectories to derive a data-driven multistep prediction model based on the Transient Predictor as a basis for our DDPC architecture.
Note that we only use the trajectories corresponding to the regular mode as this is the mode of interest that is used when the train is transporting passengers.
The prediction model uses input, output and disturbance data from the recent past (measured on-line) and predictions of the disturbance and input in the near future (the former assumed to be available through forecasts, the latter our decision variables) to predict the output trajectory in the near future. The DDPC layer then optimizes our choice of near future inputs. 
Though the train thermal dynamics involve some nonlinearities (such as the ground radiation) we restrict our attention to linear prediction models to enable faster on-line computation of the predictive controller; extending the prediction model by the inclusion of nonlinear features is an interesting direction of future work.

To account for unmeasured states such as the chassis temperature $T_{chassis}$ and inventory temperature $T_{inv}$, we consider the prediction model,
\begin{equation*}
    y_f^t = \hat{\Phi}_{p} z_p^t + \hat{\Phi}_{u} u_f^t + \hat{\Phi}_{d} d_f^t + \hat{\Phi}_{y} y_f^t,
\end{equation*}
where at time $t$, $u_f^t = [u(t+1)^\top,\hdots,u(t+T)^\top]^\top$ is the future input trajectory, $d_f^t = [d(t+1)^\top,\hdots,d(t+T)^\top]^\top$ is the future disturbance trajectory, $y_f^t = [y(t+1)^\top,\hdots,y(t+T)^\top]^\top$ is the future output trajectory and $z_p^t = [y(t-\rho+1)^\top,u(t-\rho+1)^\top,d(t-\rho+1)^\top,\hdots,y(t)^\top,u(t)^\top,d(t)^\top]^\top$ is the past output-input-disturbance trajectory 
with $\rho$ and $T$ denoting the number of lead-in data and the prediction horizon, respectively. The matrices $\hat{\Phi}_{p}$, $\hat{\Phi}_{u}$ and $\hat{\Phi}_{y}$ are block-lower-triangular and hence, define causal linear maps. In other words, these matrices are designed such that, for any $j \in \{1,\hdots,N\}$, the future output $y(t+j)$ depend only on the inputs, disturbances and outputs up to the prediction step $j-1$.
The model can be compactly represented as,
\begin{equation}
    \label{sec3B:mspm}
    y_f^t = \hat{H} 
    \begin{bmatrix}
        z_p^t \\
        u_f^t \\
        d_f^t
    \end{bmatrix},
\end{equation}
where $\hat{H}=[(I-\hat{\Phi}_{y})^{-1} \hat{\Phi}_{z}, \ (I-\hat{\Phi}_{y})^{-1} \hat{\Phi}_{u}, \ (I-\Phi_{y})^{-1} \hat{\Phi}_{d}]$.

To estimate $\hat{H}$ from the data, we define $z_f^t = [y(t+1)^\top,u(t+1)^\top,d(t+1)^\top,\hdots,y(t+T)^\top,u(t+T)^\top,d(t+T)^\top]$, allowing us to write the prediction model as,
\begin{equation*}
    y_f^t = \hat{\Phi} \begin{bmatrix} z_p^{t^\top} \\ z_f^{t^\top} \end{bmatrix}^\top, 
\end{equation*}
where $\hat{\Phi} = [\hat{\Phi}_p \ \hat{\Phi}_f]$ is the Transient Predictor \cite{moffat2024transient} with $\hat{\Phi}_f$ constructed using $\hat{\Phi}_u$, $\hat{\Phi}_y$ and $\hat{\Phi}_d$ in the obvious way. 

For each of the $N$ trajectories in our disposal, we construct the Hankel matrices,
\begin{equation*}
    \mathcal{Z}_i = 
    \begin{bmatrix}
        z_i^1 & \hspace{-0.25cm} \hdots & \hspace{-0.25cm} z_i^{n_i-T-\rho+1} \\
        \vdots & \hspace{-0.25cm} \ddots & \hspace{-0.25cm} \vdots \\
        z_i^{\rho+T} & \hspace{-0.25cm} \hdots & \hspace{-0.25cm} z_i^{n_i} \\
    \end{bmatrix},
    \mathcal{Y}_i =
    \begin{bmatrix}
        y_i^{\rho+1} & \hspace{-0.25cm} \hdots & \hspace{-0.25cm} y_i^{n_i-T+1} \\
        \vdots & \hspace{-0.25cm} \ddots & \hspace{-0.25cm} \vdots \\
        y_i^{\rho+T} & \hspace{-0.25cm} \hdots & \hspace{-0.25cm} y_i^{n_i} \\
    \end{bmatrix},
\end{equation*}
where $z_i^k = \left[y_i^{k^\top} \ u_i^{k^\top} \ d_i^{k^\top}\right]^\top$
and concatenate them to form the extended Hankel matrices,
\begin{equation*}
    \begin{aligned}
        &\begin{aligned}
            \mathcal{Z} &= 
            \begin{bmatrix}
                \mathcal{Z}_1 & \hdots & \mathcal{Z}_N \\
            \end{bmatrix} \\
            &:=
            \begin{bmatrix}
                Z_p^\top &
                Y_1^\top &
                U_1^\top &
                D_1^\top &
                \hdots &
                Y_T^\top &
                U_T^\top &
                D_T^\top
            \end{bmatrix}^\top, \\
        \end{aligned} \\
        &\mathcal{Y} = 
        \begin{bmatrix}
            \mathcal{Y}_1 & \hdots & \mathcal{Y}_N \\
        \end{bmatrix}
        :=
        \begin{bmatrix}
            Y_1^\top &
            \hdots &
            Y_T^\top
        \end{bmatrix}^\top, 
    \end{aligned}
\end{equation*}
where $Z_p \in \mathbb{R}^{9 \rho \times \mathcal{N}}$, $Y_j \in \mathbb{R}^{3 \times \mathcal{N}}$, $U_j \in \mathbb{R}^{1 \times \mathcal{N}}$ and $D_j \in \mathbb{R}^{5 \times \mathcal{N}}$ for all $j \in \{1,\hdots,T\}$ with $\mathcal{N} = \sum_{i=1}^{N} (n_i-T-\rho+1)$. 
Following \cite{moffat2024transient}, we then compute the Transient Predictor $\hat{\Phi}$ using the LQ decomposition of the extended Hankel matrix, and subsequently derive $\hat{H}$, as shown in Algorithm \ref{alg:msp}.

\begin{algorithm}
    \caption{Multistep Prediction Model}
    \label{alg:msp}
    \begin{algorithmic}[1]
        \INPUT Matrix $\mathcal{Z}$
        \OUTPUT Matrix $\hat{H}$
        \STATE $L \leftarrow LQ(\mathcal{Z})$ 
        \STATE $L^0 \leftarrow$ $L$ with the block-diagonal terms set to 0
        \STATE $L_y^0 \leftarrow$  the rows of $L^0$ corresponding to $Y_j$ for all $j \in \{1,\hdots,T\}$
        \STATE $\hat{\Phi} \leftarrow L_y^0 L^{-1}$ 
        \STATE $\hat{H} \leftarrow [(I-\hat{\Phi}_{y})^{-1} \hat{\Phi}_{z} \ (I-\hat{\Phi}_{y})^{-1} \hat{\Phi}_{u} \ (I-\hat{\Phi}_{y})^{-1} \hat{\Phi}_{d}]$ 
    \end{algorithmic}
\end{algorithm}

The Transient Predictor differs from the Subspace Predictor \cite{favoreel1999spc}, which minimizes the Frobenius norm of the Subspace Prediction \cite{moffat2025biasOfSubspace}, as the transient predictor minimizes the Frobenius Norm of the single-step prediction for the training data corresponding to each timestep in the future \cite{moffat2024transient}. The LQ decomposition achieves this minimization in a computationally efficient manner.

\subsection{Data-driven Predictive Control}
\label{sec:IIIC}

Based on the prediction model, we develop a predictive control scheme for train HVAC systems.
The proposed scheme takes the rule-based temperature $T_{rule}$ as an input and outputs the reference temperature $T_{ref}$ to be tracked by the HVAC controller as shown in Figure \ref{fig:controlSystem}.
Note that the DDPC architecture in Figure \ref{fig:controlSystem} no longer sets the reference temperature $T_{ref}$ to the rule-based temperature $T_{rule}$ as in the existing architecture in Figure \ref{fig:topology}.

\begin{figure*}[ht]
    \centering
    \scalebox{0.5}{\input{DDPC.tex}}
    \caption{Proposed control architecture.}
    \label{fig:controlSystem}
\end{figure*}
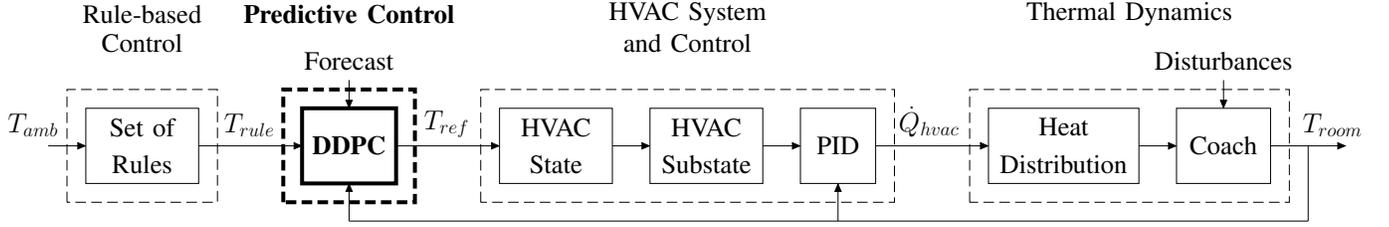

We define the future trajectories $u_f^t = [u(1|t)^\top,\hdots,u(T|t)^\top]^\top$, $y_f^t = [y(1|t)^\top,\hdots,y(T|t)^\top]^\top$, $d_f^t = [d(1|t)^\top,\hdots,d(T|t)^\top]^\top$ as well as the past and current trajectory $z_p^t = [y(t-\rho+1)^\top, u(t-\rho+1)^\top, d(t-\rho+1)^\top, \hdots, y(t-1)^\top, u(t-1)^\top, d(t-1)^\top, y(t)^\top, u(0|t)^\top, d(t)^\top]$. Note that while the current output and disturbance vectors are assumed to be measured, the current input is considered to be a decision variable.
In the spirit of MPC, we formulate and solve the optimal control problem to optimize for $u(k|t)$ for all $k \in \{0,\hdots,T\}$. We then apply the optimal control action $u^*(0|t)$ and repeat the process at the next time instant.
To predict the future thermal behavior at each time step, we consider the multistep prediction model,
\begin{equation}
    \label{eq:cons1}
    y_f^t = \hat{H} 
    \begin{bmatrix}
        z_p^t \\
        u_f^t \\
        d_f^t
    \end{bmatrix}.
\end{equation}

To ensure passenger comfort, we restrict the room temperature $T_{room}$ in the three decks within $\pm T_{max}$ around the rule-based reference temperature $T_{rule}$ as follows,
\begin{equation}
    \label{eq:cons3}
    |y(k|t)-T_{rule}(k|t)| \leq T_{max} + \epsilon(k|t) \ \forall \ k \in \{1,\hdots,T\},
\end{equation}
where $\epsilon(k|t)$ are slack variables used to maintain the feasibility of the optimal control problem even when there is not input such that $|y(k|t)-T_{rule}(k|t)| \leq T_{max}$ and satisfy,
\begin{equation}
    \label{eq:cons4}
    \epsilon(k|t) \geq 0 \ \forall \ k \in \{1,\hdots,T\}.
\end{equation}
As shown later, the slack variables are penalized in the cost to ensure passenger comfort (i.e. $|y(k|t)-T_{rule}(k|t)| \leq T_{max}$), when possible.
To match the constraints also imposed on the driver, we restrict the reference temperature $T_{ref}$ within $\pm 2 \degree C$ around the rule-based reference temperature $T_{rule}$,
\begin{equation}
    \label{eq:cons5}
     |u(k|t)-T_{rule}(k|t)| \leq 2 \ \forall \ k \in \{1,\hdots,T\}.
\end{equation}
In addition to adhering to the restrictions of the current system, together with \eqref{eq:cons3} and \eqref{eq:cons4}, this constraint also helps limit the error between the reference temperature $T_{ref}$ and the room temperature $T_{room}$. 
This helps to ensure that the output of the low-level tracking controller included in the HVAC System and Control block in Figure \ref{fig:topology},
which depends on this error, is kept within an acceptable range. In other words, this constraint indirectly implies an input constraint on the maximum allowable input heat flow rate $\dot{Q}_{hvac}$.
Finally, we restrict the change in the reference temperature $T_{ref}$ between two consecutive time steps to be at most $\Delta T_{max}$,
\begin{equation}
    \label{eq:cons6}
    \begin{aligned}
        & |u(0|t)-u(t-1)| \leq \Delta T_{max}, \\
        |u(k|t) - &u(k-1|t)| \leq \Delta T_{max}  \ \forall \ k \in \{1,\hdots,T\}.
    \end{aligned}
\end{equation}
Constraint \eqref{eq:cons6} contributes to passenger comfort by keeping an upper bound on the rate of change of the reference temperature. Moreover, it also implicitly limits the change in the input heat flow rate produced by the HVAC system.

The cost function we employ is the weighted sum of four main components. The first aims to minimize the room temperature constraint violations and takes the form,
\begin{equation*}
    J_{\epsilon} = \sum_{k=0}^T \|\epsilon(k|t)\|_2^2
\end{equation*}
To further improve the passenger comfort, we also minimize the rate of change of the reference temperature by considering the cost function,
\begin{equation*}
    J_{\Delta u} = \|u(0|t)-u(t-1)\|_2^2 + \sum_{k=1}^T \|u(k|t)-u(k-1|t)\|_2^2.
\end{equation*}

The remaining two terms of the cost function aim to minimize the energy consumption of the HVAC system. Since our controller has no access to the heat flow rates $\dot{Q}_{hvac}$, the energy consumption cannot be calculated or predicted directly. Instead we use surrogate terms that involve variables that relate to energy consumption and are available to our controller. Intuitively energy consumption is reduced if we keep the reference close to the boundary of the allowable temperatures, the upper boundary if cooling and the lower boundary when heating. Since the internal HVAC state is also not available to our controller, we use the simulation model outlined in Section \ref{sec:II} to determine which of the two boundaries the controller should be aiming for at different times over the prediction horizon.
The selection of the optimal boundary is determined using the simulator described in Section \ref{sec:II} by simulating the open-loop thermal dynamics of the coach, assuming the HVAC control block is inactive (i.e. $\dot{Q}_{hvac}=0$). The resulting open-loop temperature trajectory is then projected onto the admissible temperature interval $[T_{rule}(k|t)-T_{max}(k|t),T_{rule}(k|t)+T_{max}(k|t)]$ yielding the reference temperature $T_{opt}(k|t)$ for all $k \in \{0,\hdots,T\}$. The third term in the cost function then becomes,
\begin{equation*}
    J_{u} = \sum_{k=0}^T \|u(k|t)-T_{opt}(k|t)\|_2^2.
\end{equation*}
Including this term in the cost function allows us to select a reference trajectory that respects comfort and other constraints by achieving a trade-off with the other terms included in the cost; this would not be the case with a trivial controller that sets the reference temperature $T_{ref}(k|t)$ to the computed optimal temperature $T_{opt}(k|t)$.
Finally, we also include a term penalizing the deviation of the room temperature from the optimal bound,
\begin{equation*}
    J_{y} = \sum_{k=0}^T \|y(k|t)-T_{opt}(k|t) \textbf{1}_3\|_2^2.
\end{equation*}
where $\textbf{1}_3$ is a vector of ones of size $3 \times 1$. 

The following theorem highlights an immediate consequence of incorporating the term $J_{y}$ in the cost function.
\begin{theorem}
    Define,
    \begin{equation*}
        T_{avg}(k|t) = \frac{1}{3} \left(T_{room}^{up}(k|t)+T_{room}^{mid}(k|t)+T_{room}^{low}(k|t) \right).
    \end{equation*}
    Then, the term $J_y$ is equivalent to,
    \begin{equation}
        \label{eq:thm2}
        \sum_{k=0}^T 
        \left( 
        \begin{aligned}
            \|T_{avg}(k|t)-T_{opt}&(k|t)\|_2^2 \\
            + \| T_{room}^{up}(k|t) &- T_{avg}(k|t) \|_2^2 \\
            + \| T_{room}^{mid}&(k|t) - T_{avg}(k|t) \|_2^2 \\
            &+ \| T_{room}^{low}(k|t) - T_{avg}(k|t) \|_2^2
        \end{aligned} 
        \right).
    \end{equation}
\end{theorem}

\begin{proof}
    Recall that $y(k|t) = T_{room}(k|t)$. Therefore,
    \begin{equation*}
        \begin{aligned}
            & \|y(k|t)-T_{opt}(k|t) \textbf{1}_3\|_2^2 \\
            =
            & \|T_{room}(k|t)-T_{avg}(k|t)\textbf{1}_3+T_{avg}(k|t)\textbf{1}_3-T_{opt}(k|t)\textbf{1}_3\|_2^2 \\
            =
            & \|T_{room}(k|t)-T_{avg}(k|t)\textbf{1}_3\|_2^2 + \|T_{avg}(k|t)\textbf{1}_3-T_{opt}(k|t)\textbf{1}_3\|_2^2 \\
            +
            & 2 (T_{room}(k|t)-T_{avg}(k|t)\textbf{1}_3)^\top (T_{avg}(k|t)\textbf{1}_3-T_{opt}(k|t)\textbf{1}_3)
        \end{aligned}
    \end{equation*}
    Note that,
    \begin{equation*}
        \begin{aligned}
            &(T_{room}(k|t)-T_{avg}(k|t)\textbf{1}_3)^\top (T_{avg}(k|t)\textbf{1}_3-T_{opt}(k|t)\textbf{1}_3) \\
            =
            &(T_{avg}(k|t)-T_{opt}(k|t))(T_{room}(k|t)-T_{avg}(k|t)\textbf{1}_3)^\top \textbf{1}_3 \\
        \end{aligned}
    \end{equation*}
    Since $T_{room}(k|t)=[T_{room}^{up}(k|t) \ T_{room}^{mid}(k|t) \ T_{room}^{low}(k|t)]^\top$, then,
    \begin{equation*}
        \begin{aligned}
            &(T_{room}(k|t)-T_{avg}(k|t)\textbf{1}_3)^\top \textbf{1}_3 \\
            =
            &(T_{room}^{up}(k|t)+T_{room}^{mid}(k|t)+T_{room}^{low}(k|t)-3T_{avg}(k|t))^\top \\
            =&0.
        \end{aligned}
    \end{equation*}
    Since the above arguments hold for all $k \in \{0,\hdots,T\}$, then $J_y$ is equivalent to \eqref{eq:thm2}.
\end{proof}
The term $J_y$ therefore promotes thermal homogeneity throughout the coach. Specifically, it seeks to regulate the average temperature $T_{avg}$ towards the optimal temperature $T_{opt}$, while simultaneously minimizing the spatial variation of the room temperature $T_{room}$ across the decks. 

The resulting cost function of the optimal control problem is given by,
\begin{equation}
    J = J_y + \sigma J_u + \tau J_\epsilon + \gamma J_{\Delta u},
\end{equation}
where $\sigma$, $\tau$ and $\gamma$ are the cost function weights. 

The resulting optimal control problem is given by,
\begin{equation}
    \label{eq:OCP}
        \min J
        \text{ s.t. } \eqref{eq:cons1}, \ \eqref{eq:cons3}, \ \eqref{eq:cons4}, \ \eqref{eq:cons5}, \ \eqref{eq:cons6},
\end{equation}
with $y(k|t)$, $u(k|t)$ and $\epsilon(k|t)$ for all $k \in \{0,\hdots,T\}$ being the decision variables.

\section{Simulations}
\label{sec:IV}

In this section, we evaluate the performance of the multistep prediction model as well as the proposed DDPC architecture in simulation using the simulator described in Section \ref{sec:II}.

\subsection{Model Evaluation}

We begin by analyzing the sensitivity of the mutlistep prediction model to different parameters such as the sampling time, the number of lead-in data, $\rho$, and the prediction horizon, $T$. The performance of the model is assessed using the mean absolute error (MAE) between the actual room temperature and the predicted room temperature of a validation dataset. As shown in Figure \ref{fig:ModelEvaluation}, the model generally performs well with a mean absolute prediction error of a maximum of $0.25 \degree C$ over 60 minutes.

First, we examine the effect of the sampling time. We try three different sampling times of $1$, $5$ and $10$ minutes using the filtered data computed as discussed in Sections III.A and III.B. In all cases, we fix the lead-in time to 60 minutes, the future prediction time to 30 minutes and use the data of 9 coaches for training (six from the Winter/Spring and three from the Summer) and the data of the other three coaches for validation (two from the Winter/Spring and one from the Summer). As shown in the left panel in Figure \ref{fig:ModelEvaluation}, the prediction accuracy improves as the sampling time gets smaller. Note, however, that implementing the controller using a sampling time of $1$ minute leads to extensive communication and computation as this requires updating all weather, train and HVAC data involved in $z_p^t$ and $d_f^t$ every $1$ minute for all controlled coaches. Hence, implementing the controller with a sampling time of 1 minute might not be feasible with the low computing power of the HVAC systems, potentially hampering the widespread adoption of the proposed scheme.

To explore the effect of the number of lead-in data $\rho$ on the prediction accuracy, we fix the prediction time to 30 minutes, the sampling time to 5 minutes and keep the same division of the training and validation dataset as before. As shown in the middle panel of Figure \ref{fig:ModelEvaluation}, increasing the number of lead-in data enhances the prediction accuracy as shown in the middle plot in Figure \ref{fig:ModelEvaluation}. While the performance significantly improves from $\rho = 6$ to $\rho = 12$, the improvement is smaller as $\rho$ exceeds $12$.

Finally, to examine the effect of the prediction horizon $T$, we fix the number of lead-in data to $\rho = 6$, the sampling time to 5 minutes and keep the same division of the training and validation dataset as before. As shown in the right panel of Figure \ref{fig:ModelEvaluation}, the prediction accuracy does not significantly change at the same prediction step as the prediction horizon changes. This is due to the structure of the developed data-driven multistep prediction model.

\begin{figure*}[ht]
	\centering
    \includegraphics[scale=0.55]{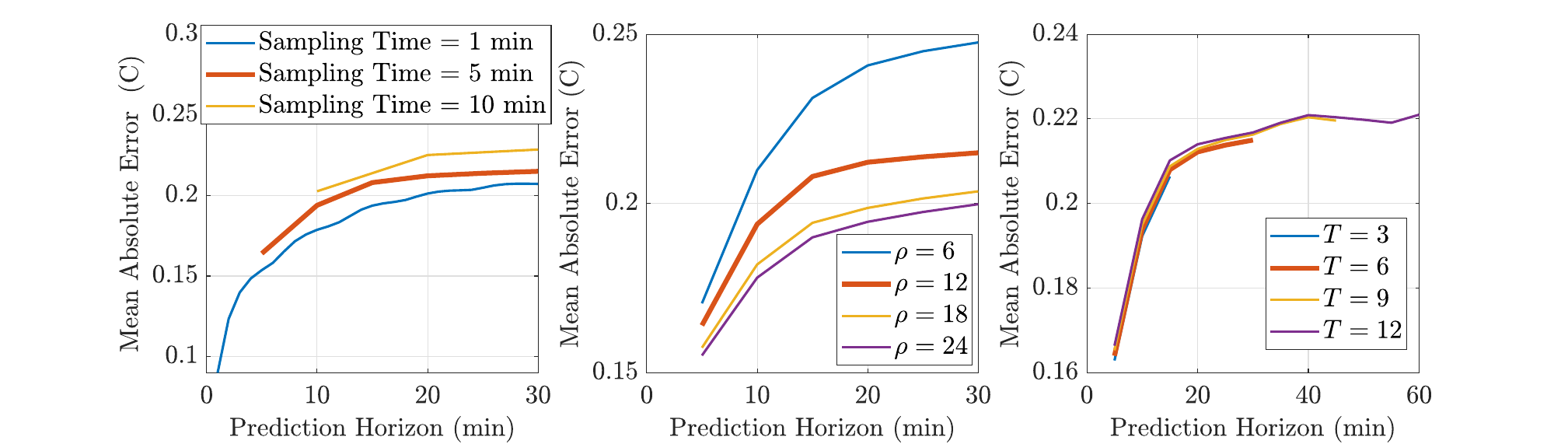}
	\caption{Data driven model evaluation (Note: The thick red line in each plot refers to the set of parameters that we will use in the proposed controller).}
	\label{fig:ModelEvaluation}
\end{figure*}

\subsection{Control Evaluation}

We now evaluate the closed-loop performance of the proposed DDPC architecture in simulation. 
We implement the controller with a sampling time of $5$ minutes and a future prediction time of $30$ minutes and use a lead-in time of $60$ minutes in the multistep prediction model. We solve the optimal control problem using MATLAB with MOSEK and YALMIP \cite{Lofberg2004}. The DDPC layer is connected in the feedback loop as shown in Figure \ref{fig:controlSystem}, with the simulator outlined in Section \ref{sec:II} replacing the ``Rule-based Control", ``HVAC System and Control" and ``Thermal Dynamics" blocks. For the external disturbances (including $T_{amb}$), we use the data of one of the coaches collected on March $28$, 2022 and March $29$, 2022. 

The closed-loop thermal behavior of the coach is shown in Figure \ref{fig:ControlEvaluation}.
The train starts its operation slightly after 04:00 on March 28 and keeps operating until around 01:00 on March 29. Since we need data for one hour to be able to use the Multistep Prediction Model, we deactivate the DDPC in the beginning of the simulation and set $T_{ref} = T_{rule}$ as shown in Figure \ref{fig:topology}.
Once the DDPC is activated at 5:15, the reference temperature moves towards the lower bound as the train needs to be heated due to the low ambient temperature and the absence of solar radiation. Around noon, the ambient temperature and solar radiation increase; as a consequence, the DDPC moves the reference temperature towards the upper bound. Finally, near the end of the day, the reference temperature moves back again to the lower bound due to the sunset and the decrease of the ambient temperature. Note that the room temperature is satisfactorily tracking the reference temperature with negligible constraint violations. If necessary, these violations can be reduced  further or eliminated by choosing slightly conservative bounds or modifying the relative weights in the cost function. 

\begin{figure*}[ht]
	\centering
    \includegraphics[scale=0.5]{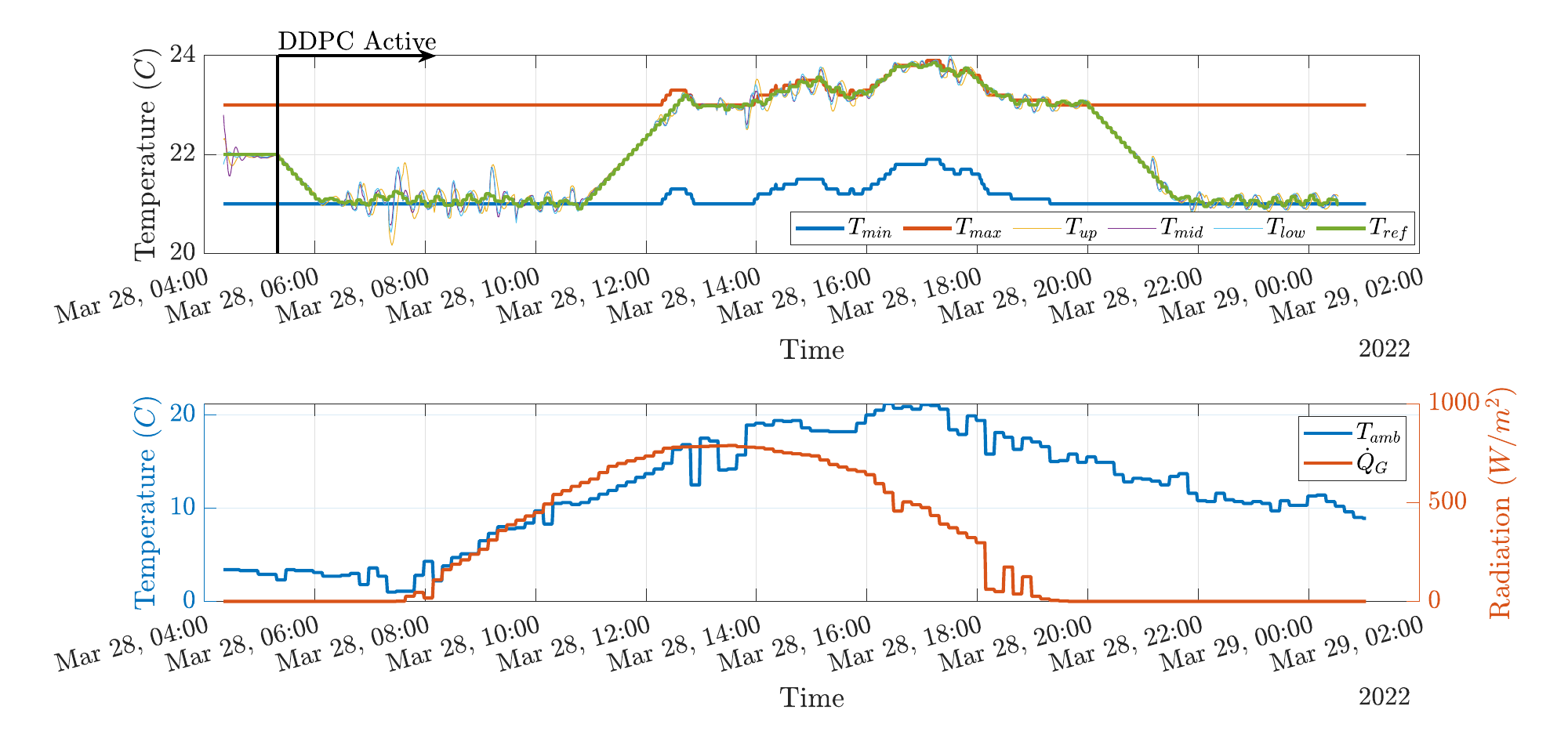}
	\caption{DDPC evaluation in simulation: The DDPC scheme moves towards the lower bound in the beginning and end of the day and towards the upper bound in the middle of the day depending on the weather conditions. Before activating the DDPC, the controller is following the midpoint.}
	\label{fig:ControlEvaluation}
\end{figure*}

\section{Experiments}
\label{sec:V}

To assess the energy efficiency of the developed DDPC scheme in practice, we performed multiple experiments on real world trains under different conditions. In each experiment, we selected two coaches in the same train. In one coach, we implemented the proposed controller (as in Figure \ref{fig:controlSystem}), while in the other, we simply set $T_{ref} = T_{rule}$ (as in Figure \ref{fig:topology}). 
We refer to the two cases as ``Activated DDPC" and ``Deactivated DDPC". 
For safety reasons, the controllers are only compared on empty non-moving trains in an outdoor setting. The experiments were conducted during the scheduled breaks of the trains. The room temperatures inside the train were measured using onboard built-in sensors, whereas the power consumption of the HVAC systems is measured using external power data loggers. The weather data for the train's location is provided in real time by SBB through Meteomatics. The optimal control problem is solved using MATLAB with MOSEK and YALMIP \cite{Lofberg2004}.

We show the results of our first experiment conducted on train 124 in Muelligen on October 9, 2024 in Figure \ref{fig:Exp1}. The performance of the coach with the ``Deactivated DDPC" is shown in the top plot. Note that the reference temperature is always in the middle between the upper and lower bounds. Notice also that, around 15:00, the temperature of the upper deck fails to follow the reference temperature. We conjecture that this is due to an internal constraint inside the HVAC system (possibly a limit on the number of times the cooling system can be switched on). Connection is lost to the HVAC system of the ``Deactivated DDPC" around 13:45. The performance of the coach with the ``Activated DDPC" is shown in the second plot; the DDPC is activated at 13:15. Although the ambient temperature and solar radiation are moderate throughout the day (see the third plot), the proposed controller drives the reference temperature towards the upper bound for saving energy. Notice that connection is lost to the HVAC system around 17:50 during this experiment. 

The energy consumption of both controllers is shown in the fourth plot; note that whenever connection to either HVAC system is lost, we set the power consumption to zero.
The coach with the DDPC activated consumes less energy in comparison to the other coach. The bar plot shows the percentage of energy savings every 30 minutes; we start computing the percentage of energy savings starting from 14:00 when the ``Activated DDPC" reaches its steady state.
Note that the ``Activated DDPC" consumes less energy not only because the temperature of the upper deck in the deactivated DDPC coach fails to follow the reference temperature. Indeed, the ``Activated DDPC" is consuming less energy throughout the experiment even when the temperature of the upper deck with the deactivated DDPC is satisfactorily following the reference temperature at the end of the experiment.

\begin{figure*}[ht]
	\centering
    \includegraphics[scale=0.5]{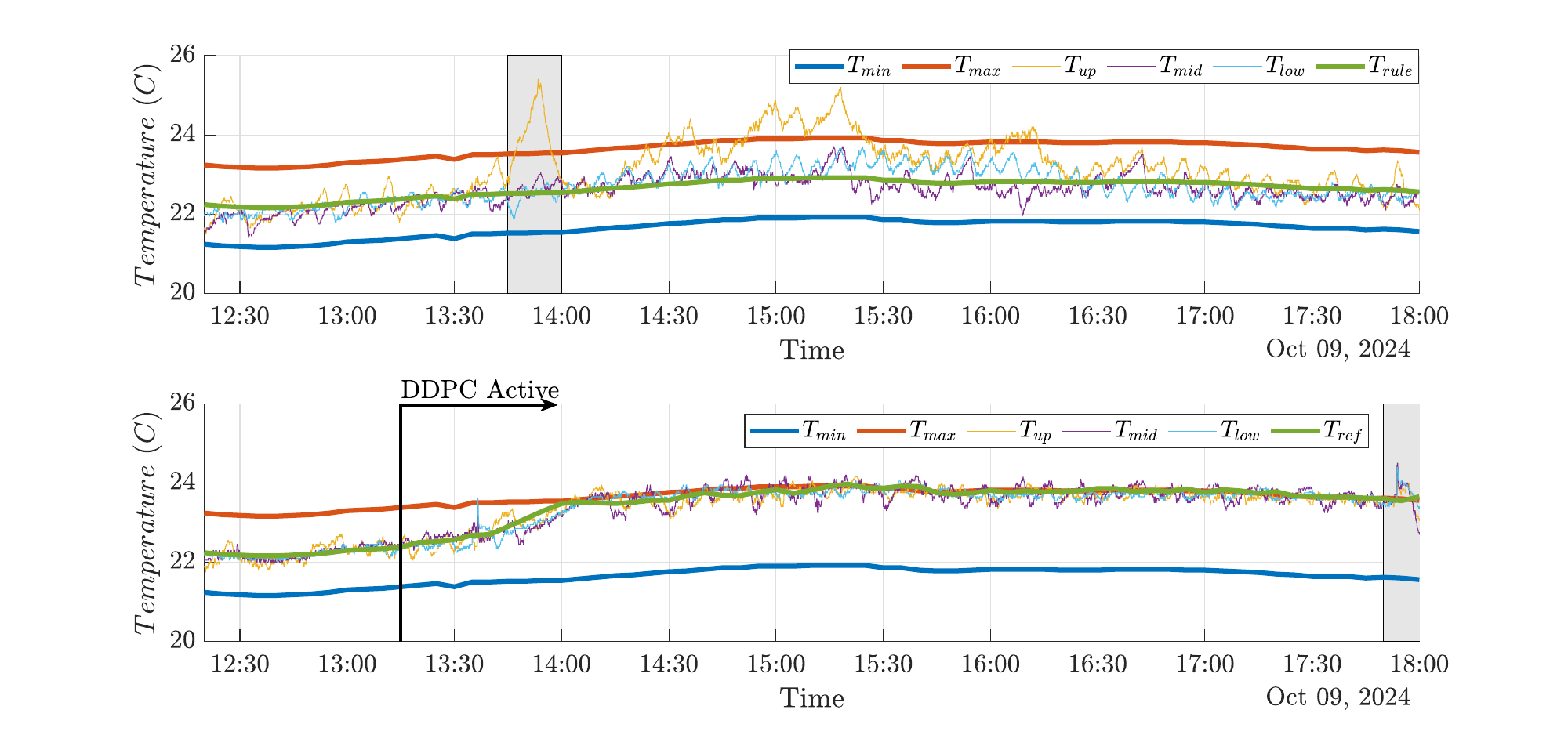}
    \includegraphics[scale=0.5]{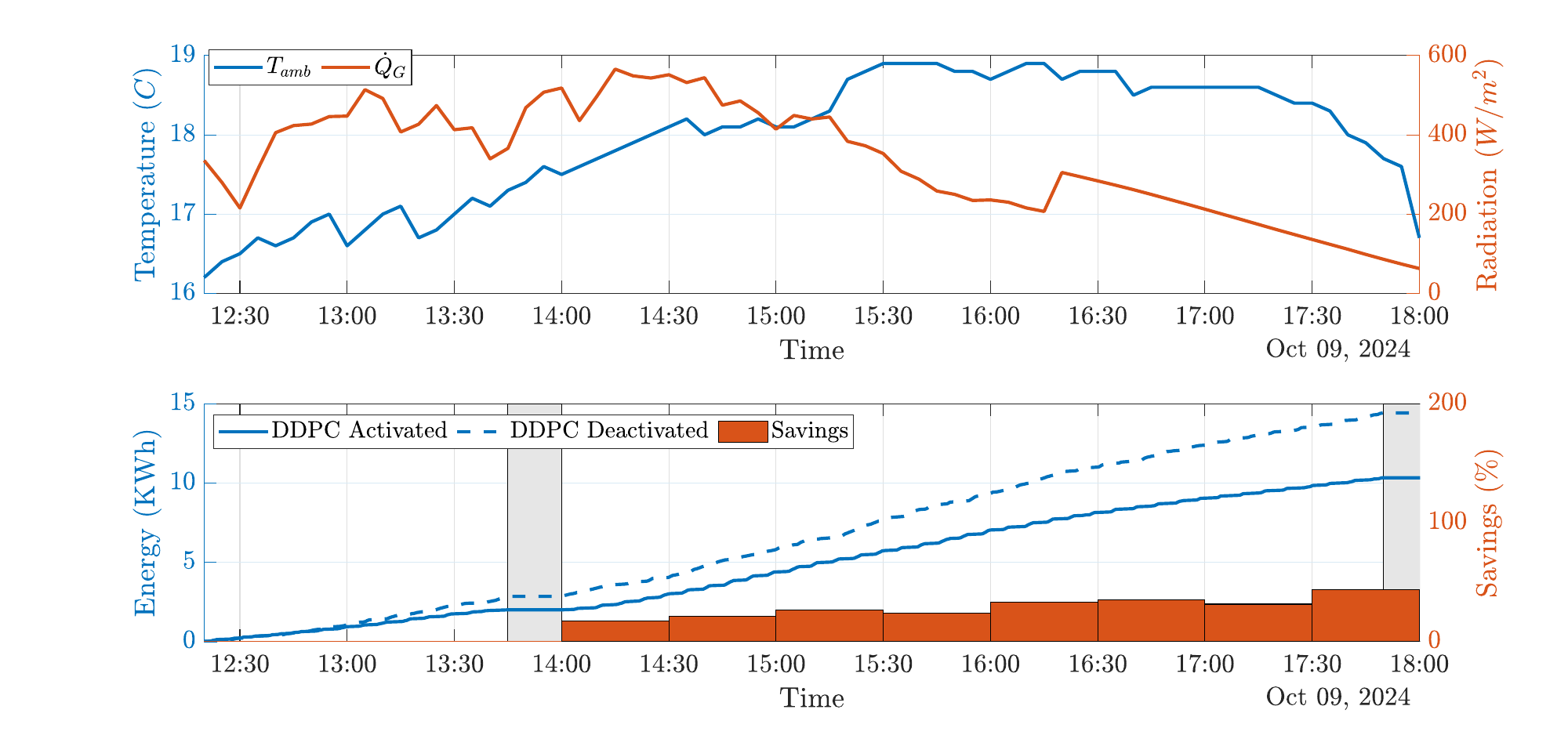}
	\caption{Cooling Experiment 1: The ``Activated DDPC" scheme aims to minimize the energy consumption while respecting the temperature bounds by moving towards the upper bound unlike the "Deactivated DDPC" scheme.}
	\label{fig:Exp1}
\end{figure*}

We show the results of our second experiment conducted on train 44 in Wiesendangen on August 15, 2024 in Figure \ref{fig:Exp2}. In this experiment, we put more emphasis on constraint violations by choosing a more conservative upper bound in the optimal control problem of the proposed controller. Note that we still show the actual bounds (i.e. $\pm 1 \degree C$) suggested by SBB  in Figure \ref{fig:Exp2}. Similar to the first experiment, the ``Deactivated DDPC" is always following the midpoint between the actual upper and lower bounds (first plot). Due to the high ambient temperature and solar radiation (third plot), the ``Activated DDPC" drives the reference temperature towards the upper bound (second plot). Note, however, that the upper bound is not reached due to the conservative constraint used by the DDPC scheme. The coach with the ``Activated DDPC" is again consuming less energy than the one with the ``Deactivated DDPC" (fourth plot).

\begin{figure*}[ht]
	\centering
    \includegraphics[scale=0.5]{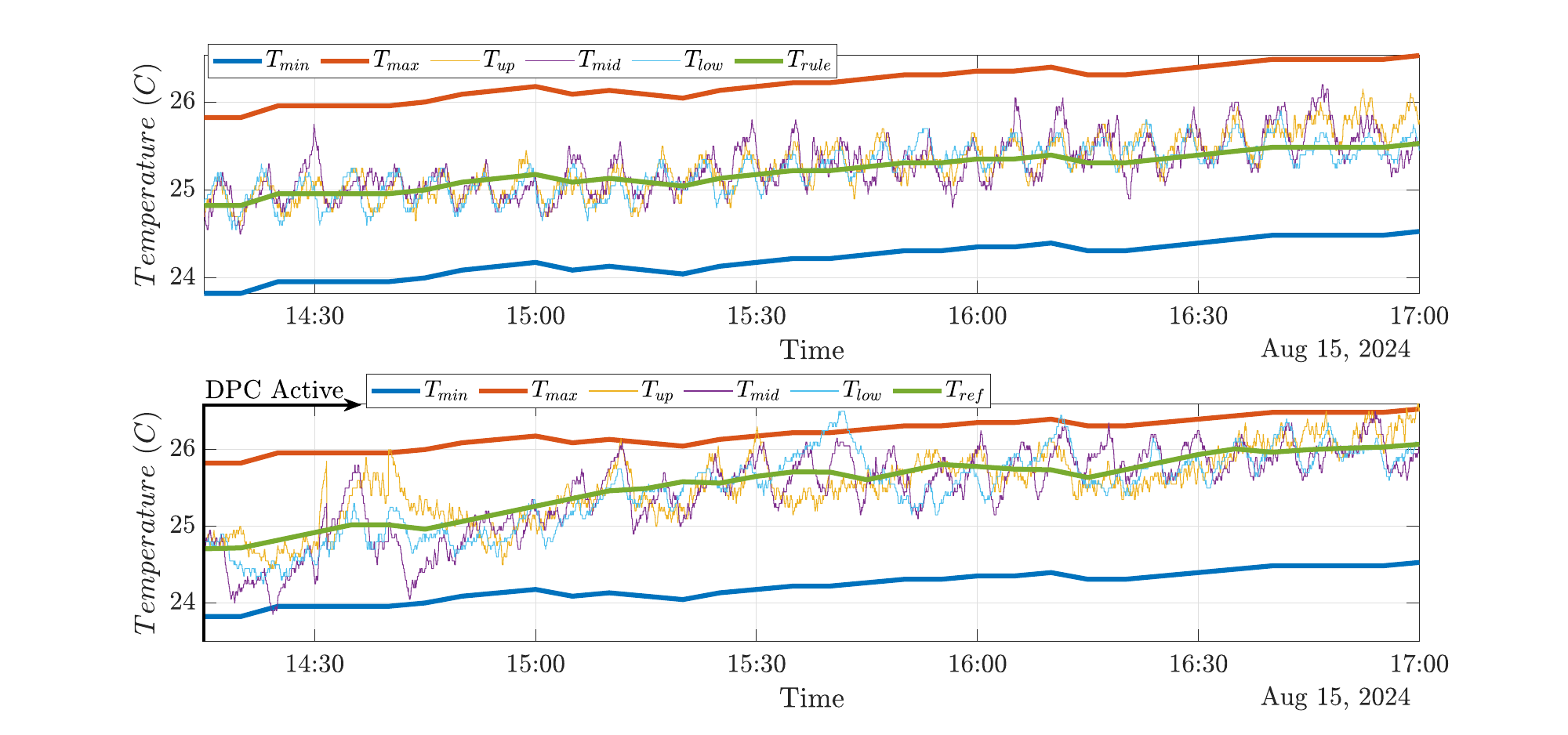}
    \includegraphics[scale=0.5]{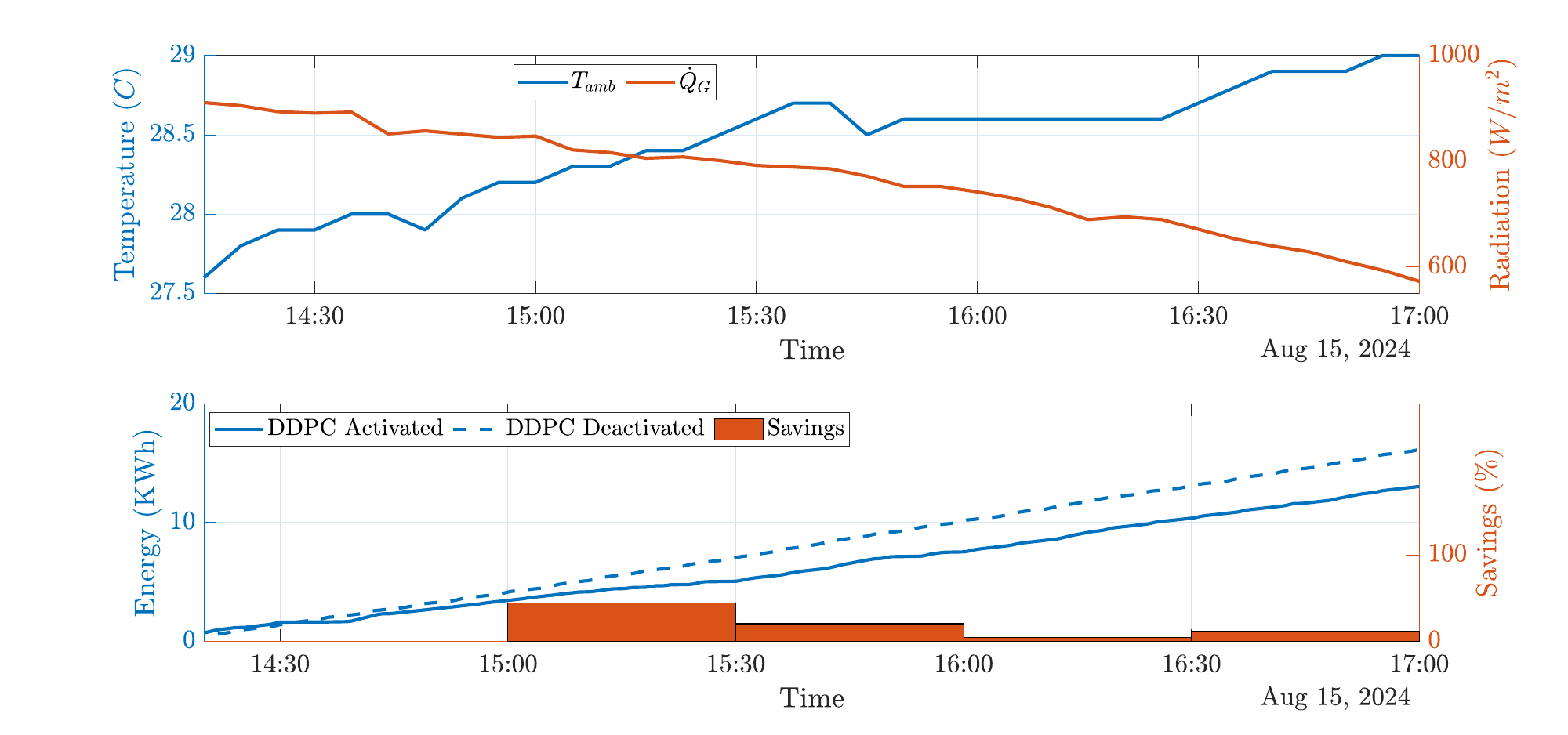}
	\caption{Cooling Experiment 2: The ``Activated DDPC" scheme consumes less energy in comparison to the “Deactivated DDPC” scheme as it approaches the actual upper bound without reaching it due to using a smaller upper bound in the optimal control problem.}
	\label{fig:Exp2}
\end{figure*}

We show the results of our third experiment performed on train 124 in Muelligen on August 22, 2024 in Figure \ref{fig:Exp3}. In this experiment, we choose the weights of the cost function to put more emphasis on energy savings. As shown in the first plot, the ``Deactivated DDPC" is again following the middle point between the bounds. Due to the high ambient temperature and solar radiation (third plot) and the increased emphasis on energy savings, the reference temperature of the ``Activated DDPC" follows the upper bound despite the significant constraint violations as shown in the second plot. Note that connection is lost to the HVAC systems of the ``Activated DDPC" around 15:15. Similar to the previous experiments, the coach with the ``Activated DDPC" is consuming less energy than the one with the ``Deactivated DDPC" (fourth plot). 

\begin{figure*}[ht]
	\centering
    \includegraphics[scale=0.5]{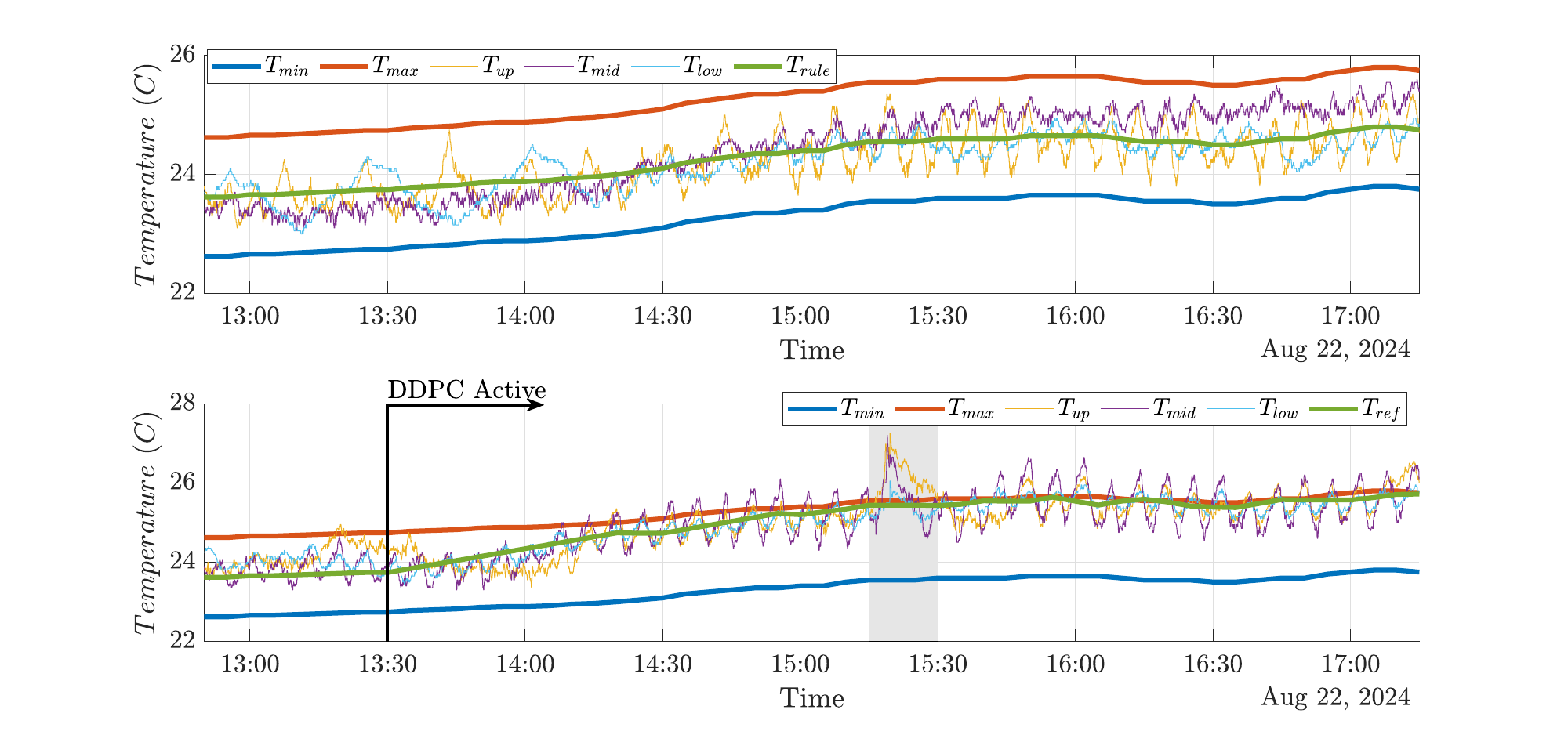}
    \includegraphics[scale=0.5]{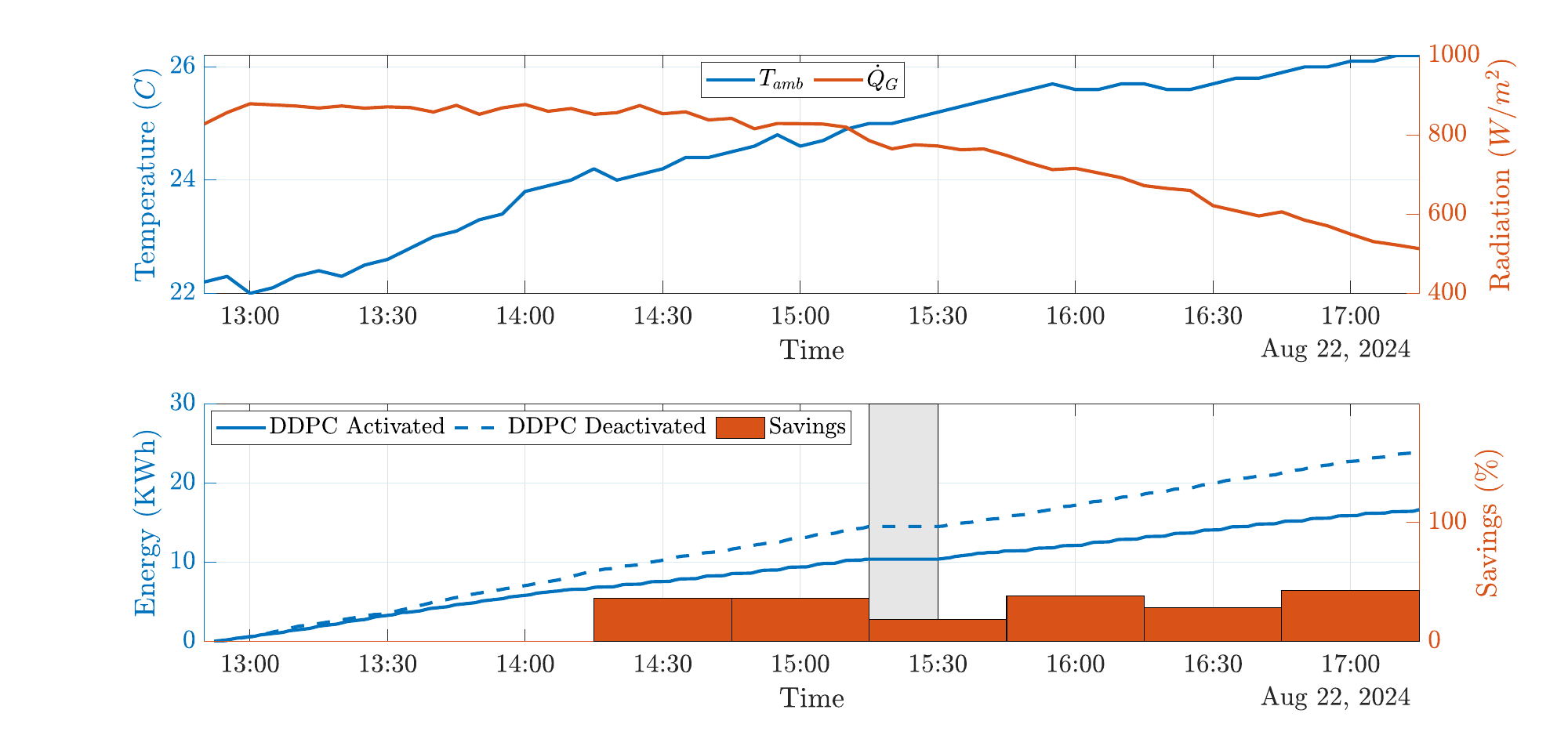}
	\caption{Cooling Experiment 3: The ``Activated DDPC" scheme minimizes the energy consumption by moving towards the upper bound and violating the temperature constraints due to the parameter selection of the optimal control problem, whereas the ``Deactivated DDPC" just follows the midpoint.}
	\label{fig:Exp3}
\end{figure*}

Finally, the results of the fourth experiment performed on train 66 in Zug on November 28, 2024 are shown in Figure \ref{fig:Exp4}. Similar to the other experiments, the ``Deactivated DDPC" is following the middle point between the bounds as shown in the first plot. On the contrary, the ``Activated DDPC" drives the reference temperature to the lower bound as shown in the second plot due to low ambient temperature and solar radiation (third plot). Note that a communication error occurred in the coach of the ``Activated DDPC" around 13:30 where the DDPC commands are overwritten over a period of approximately 30 minutes. Similar to the other experiments, the coach with the ``Activated DDPC" consumes less energy than the one with the ``Deactivated DDPC" (fourth plot).

\begin{figure*}[ht]
	\centering
    \includegraphics[scale=0.5]{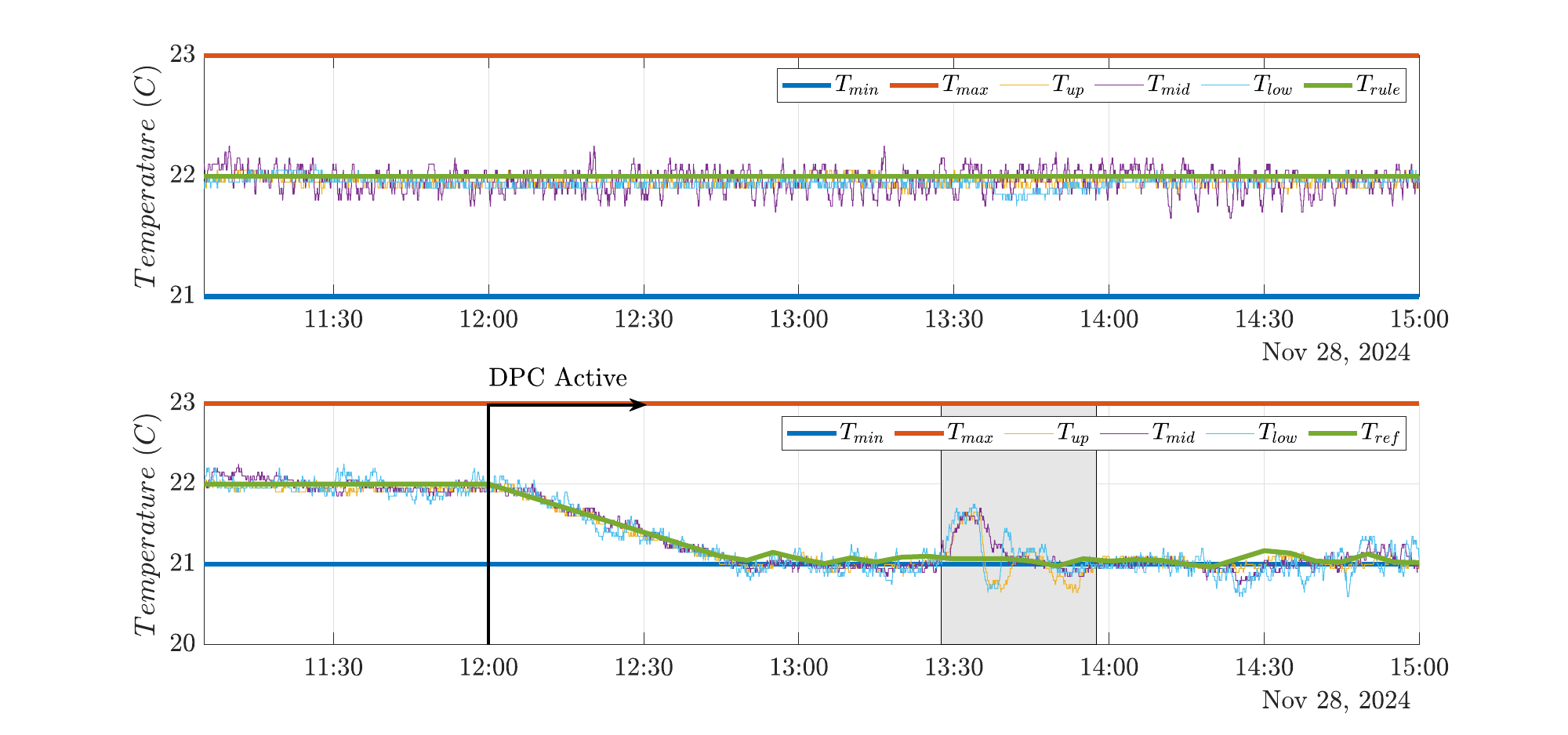}
    \includegraphics[scale=0.5]{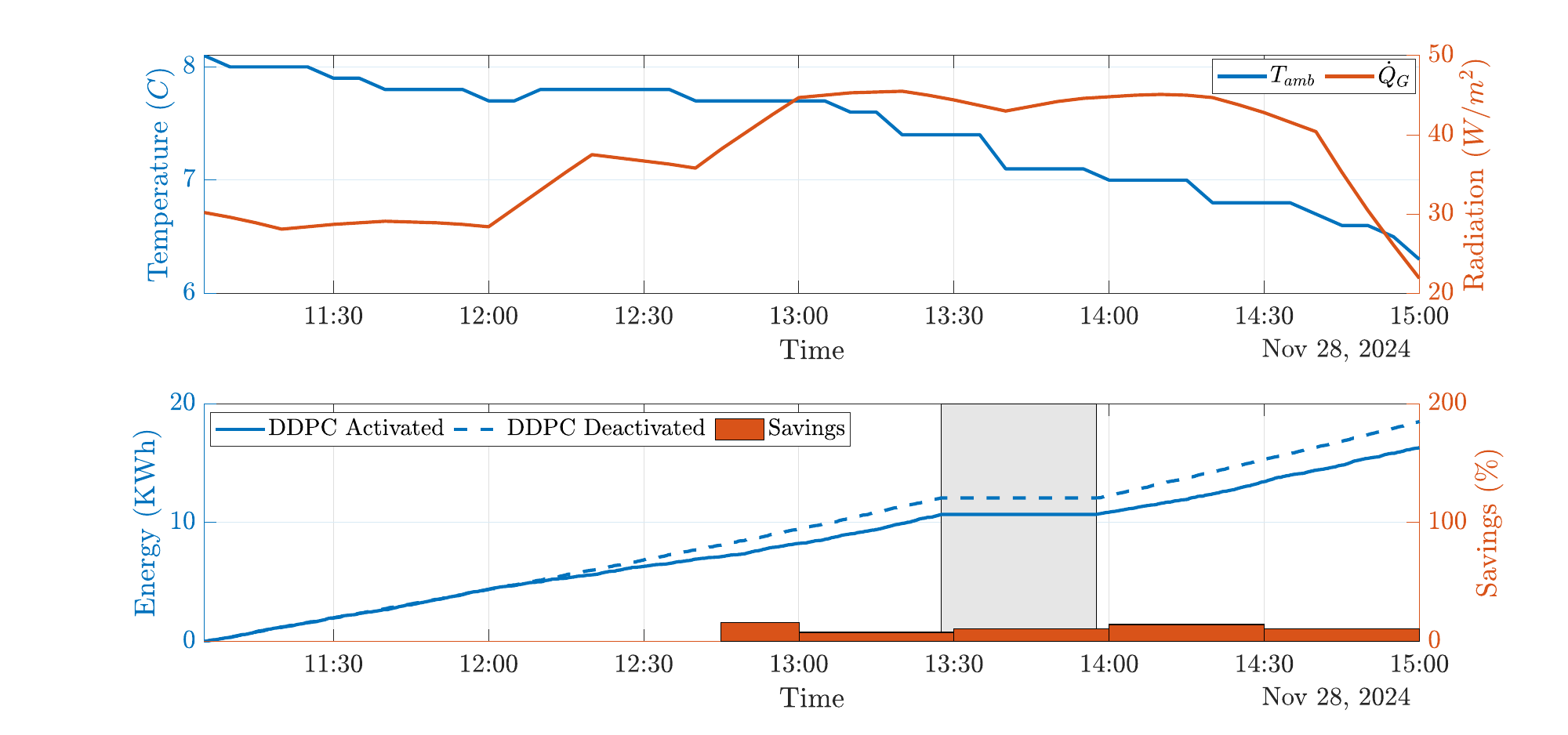}
	\caption{Heating Experiment 4: The ``Activated DDPC" scheme consumes less energy in comparison to the ``Deactivated DDPC" by moving towards the lower bound due to the winter weather conditions.}
	\label{fig:Exp4}
\end{figure*}

Finally, we present a summary of the results from the four experiments in Table \ref{tab:summary}. Specifically, we report the percentage difference in energy consumption between the coach with the "Activated DDPC" and the one with "Deactivated DDPC". This percentage is calculated for each experiment over the period when the "Activated DDPC" operates in steady state, excluding intervals affected by connection losses or communication errors. Table \ref{tab:summary} also reports the average hourly constraint violations. In the three cooling experiments (Experiments 1, 2 and 3), higher energy savings are associated with increased constraint violations. The heating experiment (Experiment 4), on the other hand, shows considerable energy savings with negligible constraint violations.

\begin{table}[]
    \centering
    \begin{tabular}{|c|c|c|c|c|}
        \hline
        Experiment & Exp. 1 & Exp. 2 & Exp. 3 & Exp. 4 \\
        \hline 
        Energy Savings ($\%$) & 28.03 & 20.11 & 34.94 & 11.65 \\
        \hline
        Average Hourly Violation ($K$) & 0.0173 & 0.0014 & 0.0609 & 0.0404 \\
        \hline
    \end{tabular}
    \caption{Summary of cooling Exp. 1, 2 \& 3 and heating Exp. 4}
    \label{tab:summary}
\end{table}

\section{Conclusion}
\label{sec:VI}

We propose a controller that aims to reduce the energy consumption of HVAC systems in trains while keeping the train temperature within prespecified acceptable bounds. Multiple experiments on separate coaches suggest that energy savings between $10\%$ and $30\%$ may be achievable with negligible constraint violations by appropriately tuning the control parameters. Future work includes deploying the developed DDPC architecture in more complex and dynamic scenarios, such as occupied trains. These scenarios require a more rigorous methodology to enable meaningful comparisons, accounting for inherent variability between coaches (e.g., differences in occupancy). Another direction is to evaluate the controller’s impact on various passenger comfort-related key performance indicators. 
Incorporating nonlinear prediction models—as done in related work on thermal energy management in buildings \cite{di2022physically}—is also promising. Finally, enhancing weather and occupancy forecasts, as well as exploring robust or stochastic formulations to handle forecast uncertainties, represents an important direction for future research. 

\section*{Acknowledgments}

We would like to thank Dr. Raffaele Soloperto and Francesco Micheli for the fruitful discussions on technical topics.  We are also grateful for Matthias Tuchschmid, Elodie Duliscouet, Michael Laszlo, Anina Döbeli, Francesco Scotto, Dr. Ralf Hofer, Dr. Achim Zanker, Stephan Geiger, Alex Müri, Sepp Stark, Moritz Kupper, Andreas Fuchs and Jan Gasser for their insightful support, valuable discussions and for providing data and information on the trains.

\bibliographystyle{ieeetr}
\bibliography{references}

\vfill

\end{document}

%% file: Simulator.tex
\begin{tikzpicture}

\draw[black, dash pattern= {on 8pt off 4pt}] (1.5,0.5) rectangle (5.5,3.5); 
\huge
\node at (3.5,5.5) [] {\textcolor{black}{Rule-based}}; 
\node at (3.5,4.75) [] {\textcolor{black}{Control}}; 
\draw[black] (2,1) rectangle (5,3); 
\huge
\node at (3.5,2.5) [] {\textcolor{black}{Set of}};
\node at (3.5,1.5) [] {\textcolor{black}{Rules}};
\draw[-triangle 45] (1,2) -- (2,2) node[left, above, xshift=-40pt] {$T_{amb}$};
\draw[-triangle 45] (5,2) -- (13,2) node[left, above, xshift=-115pt] {$T_{rule} = T_{ref}$};

\draw[black, dash pattern= {on 8pt off 4pt}] (12.5,0.5) rectangle (23.5,3.5); 
\huge
\node at (18,5.5) [] {\textcolor{black}{HVAC System}};
\node at (18,4.75) [] {\textcolor{black}{and Control}};
\draw[black] (13,1) rectangle (16,3); 
\huge
\node at (14.5,2.5) [] {\textcolor{black}{HVAC}};
\node at (14.5,1.5) [] {\textcolor{black}{State}};
\draw[-triangle 45] (16,2) -- (17,2); 
\draw[black] (17,1) rectangle (20,3); 
\huge
\node at (18.5,2.5) [] {\textcolor{black}{HVAC}};
\node at (18.5,1.5) [] {\textcolor{black}{Substate}};
\draw[-triangle 45] (20,2) -- (21,2); 
\draw[black] (21,1) rectangle (23,3); 
\huge
\node at (22,2) [] {\textcolor{black}{PID}};
\draw[-triangle 45] (23,2) -- (26,2) node[left, above, xshift=-45pt] {$\dot{Q}_{hvac}$};

\draw[black, dash pattern= {on 8pt off 4pt}] (25.5,0.5) rectangle (34,3.5); 
\huge
\node at (29.75,5.5) [] {\textcolor{black}{Thermal Dynamics}}; 
\draw[black] (26,1) rectangle (30,3); 
\huge
\node at (28,2.5) [] {\textcolor{black}{Heat}};
\node at (28,1.5) [] {\textcolor{black}{Distribution}};
\draw[-triangle 45] (30,2) -- (31,2); 
\draw[black] (31,1) rectangle (33.5,3); 
\huge
\node at (32.25,2) [] {\textcolor{black}{Coach}};
\draw[-triangle 45] (33.5,2) -- (35.5,2) node[right, above, xshift=-10pt] {$T_{room}$};
\draw[-] (34.5,2) -- (34.5,0);
\draw[-] (34.5,0) -- (22,0);
\draw[-triangle 45] (22,0) -- (22,1);
\draw[-triangle 45] (32.25,3.75) -- (32.25,3);
\node at (32.25,4.25) [] {\textcolor{black}{Disturbances}};

\end{tikzpicture}

%% file: ThermalDynamics.tex
\begin{tikzpicture}

\draw[ultra thick, dash pattern=on 10pt off 5pt, ->] (10.5, -1) -- (10.5, 0);
\draw[ultra thick, dash pattern=on 10pt off 5pt, ->] (11.5, -1) -- (11.5, 6);
\draw[ultra thick, dash pattern=on 10pt off 5pt, ->] (12.5, -1) -- (12.5, 2);
\draw[ultra thick, dash pattern=on 10pt off 5pt, ->] (23.5, -1) -- (23.5, 2);
\draw[ultra thick, dash pattern=on 10pt off 5pt, ->] (24.5, -1) -- (24.5, 0);
\draw[ultra thick, dash pattern=on 10pt off 5pt, ->] (25.5, -1) -- (25.5, 6);
\draw[ultra thick, dash pattern=on 10pt off 5pt, ->] (10.5, 14) -- (10.5, 7);
\draw[ultra thick, dash pattern=on 10pt off 5pt, ->] (11.5, 14) -- (11.5, 13);
\draw[ultra thick, dash pattern=on 10pt off 5pt, ->] (12.5, 14) -- (12.5, 11);
\draw[ultra thick, dash pattern=on 10pt off 5pt, ->] (23.5, 14) -- (23.5, 11);
\draw[ultra thick, dash pattern=on 10pt off 5pt, ->] (24.5, 14) -- (24.5, 7);
\draw[ultra thick, dash pattern=on 10pt off 5pt, ->] (25.5, 14) -- (25.5, 13);

\draw[ultra thick, dashed, <->] (26.5, 1) -- (26.5, 6);
\draw[ultra thick, dashed, <->] (26.5, 7) -- (26.5, 12);
\draw[ultra thick, dashed, <->] (27.5, 1) -- (27.5, 12);

\draw[ultra thick, <->] (12.5, 3) -- (12.5, 6);
\draw[ultra thick, <->] (12.5, 7) -- (12.5, 10);

\draw[ultra thick, dash pattern=on 15pt off 10pt on 5pt off 10pt, <->] (5, 0.25) -- (8, 0.25);
\draw[ultra thick, dash pattern=on 15pt off 10pt on 5pt off 10pt, <->] (5, 2.5) -- (12, 2.5);
\draw[ultra thick, dash pattern=on 15pt off 10pt on 5pt off 10pt, <->] (5, 6.5) -- (8, 6.5);
\draw[ultra thick, dash pattern=on 15pt off 10pt on 5pt off 10pt, <->] (5, 10.5) -- (12, 10.5);
\draw[ultra thick, dash pattern=on 15pt off 10pt on 5pt off 10pt, <->] (5, 12.75) -- (8, 12.75);

\draw[ultra thick, <->] (5, 0.75) -- (8, 0.75);
\draw[ultra thick, <->] (9.5, 1) -- (9.5, 6);
\draw[ultra thick, <->] (9.5, 7) -- (9.5, 12);
\draw[ultra thick, <->] (5, 12.25) -- (8, 12.25);
\draw[ultra thick, <->] (13, 6.5) -- (15, 6.5);
\draw[ultra thick, <->] (18, 3) -- (18, 4);
\draw[ultra thick, <->] (18, 9) -- (18, 10);
\draw[ultra thick, <->] (22.5, 1) -- (22.5, 2);
\draw[ultra thick, <->] (22.5, 3) -- (22.5, 6);
\draw[ultra thick, <->] (22.5, 7) -- (22.5, 10);
\draw[ultra thick, <->] (22.5, 11) -- (22.5, 12);

\draw[black,fill=white] (0,0) rectangle (5,13); 
\huge
\node at (2.5,6.5) [] {\textcolor{black}{Environment}};

\draw[black,fill=white] (8,0) rectangle (28,1); 
\huge
\node at (18,0.5) [] {\textcolor{black}{Lower Chassis}};
\draw[black,fill=white] (22,6) rectangle (27,7); 
\huge
\node at (24.5,6.5) [] {\textcolor{black}{Middle Floor}};
\draw[black,fill=white] (8,12) rectangle (28,13); 
\huge
\node at (18,12.5) [] {\textcolor{black}{Upper Chassis}};

\draw[black,fill=white] (12,2) rectangle (24,3); 
\huge
\node at (18,2.5) [] {\textcolor{black}{Lower Deck}};
\draw[black,fill=white] (8,6) rectangle (13,7); 
\huge
\node at (10.5,6.5) [] {\textcolor{black}{Middle Deck}};
\draw[black,fill=white] (12,10) rectangle (24,11); 
\huge
\node at (18,10.5) [] {\textcolor{black}{Upper Deck}};

\draw[black,fill=white] (15,4) rectangle (20,5); 
\huge
\node at (17.5,4.5) [] {\textcolor{black}{Lower Inventory}};
\draw[black,fill=white] (15,6) rectangle (20,7); 
\huge
\node at (17.5,6.5) [] {\textcolor{black}{Middle Inventory}};
\draw[black,fill=white] (15,8) rectangle (20,9); 
\huge
\node at (17.5,8.5) [] {\textcolor{black}{Upper Inventory}};

\draw[black,fill=white] (9,-2) rectangle (17,-1); 
\huge
\node at (13,-1.5) [] {\textcolor{black}{HVAC}};
\draw[black,fill=white] (19,14) rectangle (27,15); 
\huge
\node at (23,14.5) [] {\textcolor{black}{HVAC}};
\draw[black,fill=white] (9,14) rectangle (17,15); 
\huge
\node at (13,14.5) [] {\textcolor{black}{Floor/Wall Heating}};
\draw[black,fill=white] (19,-2) rectangle (27,-1); 
\huge
\node at (23,-1.5) [] {\textcolor{black}{Floor/Wall Heating}};


\node at (3,-3) [] {\textcolor{black}{Conduction}};
\draw[ultra thick, dashed, <->] (5, -3) -- (7, -3);
\node at (10,-3) [] {\textcolor{black}{Convection}};
\draw[ultra thick, <->] (12,-3) -- (14,-3);
\node at (17,-3) [] {\textcolor{black}{Radiation}};
\draw[ultra thick, dash pattern=on 15pt off 10pt on 5pt off 10pt, <->] (19, -3) -- (21, -3);
\node at (24,-3) [] {\textcolor{black}{Input}};
\draw[ultra thick, dotted, ->] (25, -3) -- (27, -3);

\end{tikzpicture}

%% file: DDPC.tex
\begin{tikzpicture}

\draw[black, dash pattern= {on 8pt off 4pt}] (1.5,0.5) rectangle (5.5,3.5); 
\huge
\node at (3.5,5.5) [] {\textcolor{black}{Rule-based}}; 
\node at (3.5,4.75) [] {\textcolor{black}{Control}}; 
\draw[black] (2,1) rectangle (5,3); 
\huge
\node at (3.5,2.5) [] {\textcolor{black}{Set of}};
\node at (3.5,1.5) [] {\textcolor{black}{Rules}};
\draw[-triangle 45] (1,2) -- (2,2) node[left, above, xshift=-40pt] {$T_{amb}$};
\draw[-triangle 45] (5,2) -- (7.75,2) node[left, above, xshift=-40pt] {$T_{rule}$};

\draw[black, dash pattern= {on 8pt off 4pt}, line width=3pt] (7.25,0.5) rectangle (10.75,3.5); 
\huge
\node at (9,5.5) [] {\textcolor{black}{\textbf{Predictive Control}}}; 
\draw[black, line width=3pt] (7.75,1) rectangle (10.25,3); 
\huge
\node at (9,2) [] {\textcolor{black}{\textbf{DDPC}}};
\draw[-triangle 45] (10.25,2) -- (13,2) node[left, above, xshift=-40pt] {$ T_{ref}$};
\draw[-triangle 45] (9,3.75) -- (9,3);
\node at (9,4.25) [] {\textcolor{black}{Forecast}};

\draw[black, dash pattern= {on 8pt off 4pt}] (12.5,0.5) rectangle (23.5,3.5); 
\huge
\node at (18,5.5) [] {\textcolor{black}{HVAC System}};
\node at (18,4.75) [] {\textcolor{black}{and Control}};
\draw[black] (13,1) rectangle (16,3); 
\huge
\node at (14.5,2.5) [] {\textcolor{black}{HVAC}};
\node at (14.5,1.5) [] {\textcolor{black}{State}};
\draw[-triangle 45] (16,2) -- (17,2); 
\draw[black] (17,1) rectangle (20,3); 
\huge
\node at (18.5,2.5) [] {\textcolor{black}{HVAC}};
\node at (18.5,1.5) [] {\textcolor{black}{Substate}};
\draw[-triangle 45] (20,2) -- (21,2); 
\draw[black] (21,1) rectangle (23,3); 
\huge
\node at (22,2) [] {\textcolor{black}{PID}};
\draw[-triangle 45] (23,2) -- (26,2) node[left, above, xshift=-45pt] {$\dot{Q}_{hvac}$};

\draw[black, dash pattern= {on 8pt off 4pt}] (25.5,0.5) rectangle (34,3.5); 
\huge
\node at (29.75,5.5) [] {\textcolor{black}{Thermal Dynamics}}; 
\draw[black] (26,1) rectangle (30,3); 
\huge
\node at (28,2.5) [] {\textcolor{black}{Heat}};
\node at (28,1.5) [] {\textcolor{black}{Distribution}};
\draw[-triangle 45] (30,2) -- (31,2); 
\draw[black] (31,1) rectangle (33.5,3); 
\huge
\node at (32.25,2) [] {\textcolor{black}{Coach}};
\draw[-triangle 45] (33.5,2) -- (35.5,2) node[right, above, xshift=-10pt] {$T_{room}$};
\draw[-] (34.5,2) -- (34.5,0);
\draw[-] (34.5,0) -- (9,0);
\draw[-triangle 45] (22,0) -- (22,1);
\draw[-triangle 45] (9,0) -- (9,1);
\draw[-triangle 45] (32.25,3.75) -- (32.25,3);
\node at (32.25,4.25) [] {\textcolor{black}{Disturbances}};

\end{tikzpicture}